\newif\ifstoc
\stoctrue 
\stocfalse

\documentclass[letterpaper,11pt]{article}
\usepackage{typearea}
\typearea{14}
\usepackage[compact]{titlesec}
\usepackage{theorem}

\ifstoc
  \usepackage{times}
  \usepackage{wrapfig}
\else
\fi

\newcommand{\stocoption}[2]{{\ifstoc%
#1%
\else%
#2%
\fi}}
\newcommand{\full}[1]{\ifstoc\else#1\fi}
\newcommand{\short}[1]{\ifstoc#1\fi}

\usepackage{latexsym,graphicx}
\usepackage{amsmath,amssymb,enumerate}
\usepackage{boxedminipage,wrapfig}
\usepackage{xspace}
\usepackage{bm}
\usepackage{ifpdf}
\usepackage{color}
\usepackage[compact]{titlesec}

\usepackage{mathrsfs}

\allowdisplaybreaks

\definecolor{Darkblue}{rgb}{0,0,0.4}
\definecolor{Brown}{cmyk}{0,0.81,1.,0.60}
\definecolor{Purple}{cmyk}{0.45,0.86,0,0}
\newcommand{\mydriver}{hypertex}
\ifpdf
 \renewcommand{\mydriver}{pdftex}
\fi
\usepackage[pdfusetitle,breaklinks,\mydriver]{hyperref}
\hypersetup{colorlinks=true,
            citebordercolor={.6 .6 .6},linkbordercolor={.6 .6 .6},%
citecolor=Darkblue,urlcolor=black,linkcolor=Darkblue,pagecolor=black}
\newcommand{\lref}[2][]{\hyperref[#2]{#1~\ref*{#2}}}


\newtheorem{theorem}{Theorem}[section]
\newtheorem{definition}[theorem]{Definition}

\newtheorem{lemma}[theorem]{Lemma}
\newtheorem{fact}[theorem]{Fact}

\newtheorem{claim}[theorem]{Claim}

\newtheorem{corollary}[theorem]{Corollary}

\newcommand{\Fs}{\mathscr{F}^{\star}}
\newcommand{\Fss}{\mathscr{F}^{\star\star}}
\newcommand{\Tss}{\mathscr{T}^{\star\star}}
\newcommand{\graph}[1]{{G^{(#1)}}}
\newcommand{\Kss}{\mathscr{K}^{\star\star}}
\newcommand{\fcl}[1]{{\mathscr{F}^{(#1)}}}
\newcommand{\tcl}[1]{{T^{(#1)}}}
\newcommand{\act}{{\mathsf{A}}}

\newenvironment{proofof}[1]{

\noindent{\bf Proof of {#1}:}}
{\hfill$\blacksquare$

}

\newcommand{\junk}[1]{}
\newcommand{\ignore}[1]{}

\newcommand{\sse}{\subseteq}

\newcommand{\C}{{\mathscr{C}}}
\newcommand{\E}{{\mathscr{E}}}
\newcommand{\A}{{\mathscr{A}}}
\newcommand{\B}{{\mathscr{B}}}
\newcommand{\D}{{\mathscr{D}}}
\newcommand{\ti}{j}
\newcommand{\F}{{\mathscr{F}}}
\newcommand{\G}{{\mathscr{G}}}

\newcommand{\I}{{\mathscr{I}}}
\newcommand{\M}{{\mathcal{M}}}
\newcommand{\R}{{\mathcal{R}}}

\newcommand{\ts}{\textstyle}
\renewcommand{\sp}{{\hspace*{0.1 in}}}

\newcommand{\rom}[1]{\uppercase\expandafter{\romannumeral #1}}

\newcommand{\OPT}{\ensuremath{\mathsf{opt}\xspace}}

\newcommand{\cost}{\ensuremath{\mathsf{cost}\xspace}}
\newcommand{\del}{\ensuremath{\mathsf{Del}\xspace}}
\newcommand{\length}{\ensuremath{\mathsf{length}\xspace}}
\newcommand{\pl}{\ensuremath{\mathsf{\psi}\xspace}}
\newcommand{\clus}[1]{{\C}^{(#1)}}
\newcommand{\stagecl}[1]{{\C}^{#1}}
\newcommand{\clusr}[2]{{\C}^{(#1)}_{#2}}

\newcommand{\activecl}[1]{\smash{{\C}^{(#1)}_{{\small active}}}}

\newcommand{\Ts}{T^\star}
\newcommand{\alive}{\ensuremath{\mathsf{alive}\xspace}}
\renewcommand{\time}{\ensuremath{\mathsf{time}\xspace}}
\newcommand{\charge}{\ensuremath{\mathsf{charge}\xspace}}
\newcommand{\head}{\ensuremath{\mathsf{head}\xspace}}

\newcommand{\one}[2]{{\bf 1}[#1,#2]}

\newcommand{\timedglut}{\ensuremath{\texttt{TimedGlut}}\xspace}
\newcommand{\level}{\ensuremath{\mathsf{level}}\xspace}
\newcommand{\leader}{\ensuremath{\mathsf{leader}}\xspace}

\newcommand{\forestone}[1]{\smash{F_1^{(#1)}}}
\newcommand{\foresttwo}[1]{\smash{F_2^{(#1)}}}
\newcommand{\moatsone}[1]{{\mathfrak{M}_1^{(#1)}}}
\newcommand{\stageg}[1]{{H^{(#1)}}}
\newcommand{\pairs}[1]{{{\cal P}^{(#1)}}}
\newcommand{\pairsgood}[1]{{{\cal P}_G^{(#1)}}}

\newcommand{\countbad}[1]{\smash{p_B^{(#1)}}}
\newcommand{\pairsbad}[1]{\smash{{\cal P}_B^{(#1)}}}
\newcommand{\pairsdrop}[1]{\smash{{\cal P}_X^{(#1)}}}

\newcommand{\cee}{2}
\newcommand{\ceesq}{4}
\newcommand{\ceesqplusone}{5}

\newcommand{\TPD}{\texttt{TimedPD}\xspace}
\newcommand{\updateforest}{\texttt{UpdateForest}\xspace}
\newcommand{\gammatg}{\gamma_{{\scriptscriptstyle TG}}}

\newcounter{note}[section]

\newcommand{\qedsymb}{\hfill{\rule{2mm}{2mm}}}
\newenvironment{proof}{\begin{trivlist} \item[\hspace{\labelsep}{\bf
\noindent Proof.\/}] }{\qedsymb\end{trivlist}}%

\newcommand{\initOneLiners}{%
    \setlength{\itemsep}{0pt}
    \setlength{\parsep }{0pt}
    \setlength{\topsep }{0pt}
}
\newenvironment{OneLiners}[1][\ensuremath{\bullet}]
    {\begin{list}
        {#1}
        {\initOneLiners}}
    {\end{list}}

\newcommand{\squishlist}{
 \begin{list}{$\bullet$}
  { \setlength{\itemsep}{0pt}
     \setlength{\parsep}{3pt}
     \setlength{\topsep}{3pt}
     \setlength{\partopsep}{0pt}
     \setlength{\leftmargin}{1.5em}
     \setlength{\labelwidth}{1em}
     \setlength{\labelsep}{0.5em} } }

\newcommand{\squishend}{
  \end{list}  }

\newcommand{\stf}{Steiner forest\xspace}
\newcommand{\width}{\mathsf{width}\xspace}
\newcommand{\gs}{{strict}\xspace}
\newcommand{\us}{{uni-strict}\xspace}

\begin{document}
\title{Greedy Algorithms for Steiner Forest}

\author{
Anupam Gupta\thanks{Computer Science Department, Carnegie Mellon
    University, Pittsburgh, PA 15213, USA. Research partly supported by
    NSF awards CCF-1016799 and CCF-1319811.}
\and Amit Kumar\thanks{Dept. of Computer Science and Engg., IIT Delhi,
  India 110016.}
}
\date{}
\maketitle
\thispagestyle{empty}
\stocoption{}{
\thispagestyle{empty}
\setcounter{page}{0}
}

\begin{abstract}
  In the Steiner Forest problem, we are given terminal pairs $\{s_i,
  t_i\}$, and need to find the cheapest subgraph which connects each of
  the terminal pairs together. In 1991, Agrawal, Klein, and Ravi, and
  Goemans and Williamson gave primal-dual constant-factor approximation
  algorithms for this problem; until now the only constant-factor
  approximations we know are via linear programming relaxations.

  \medskip\noindent
  In this paper, we consider the following greedy algorithm:
  \begin{quote}
    \emph{Given terminal pairs in a metric space, a terminal is active
      if its distance to its partner is non-zero. Pick the two closest
      active terminals (say $s_i, t_j$), set the distance between
      them to zero, and buy a path connecting them.  Recompute the
      metric, and repeat.}
  \end{quote}
  It has long been open to analyze this greedy algorithm. Our main
  result: this algorithm is a constant-factor approximation.

  \medskip\noindent
  We use this algorithm to give new, simpler constructions of
  cost-sharing schemes for Steiner forest. In particular, the first
  ``\gs'' cost-shares for this problem implies a very simple
  combinatorial sampling-based algorithm for stochastic Steiner forest.
\end{abstract}

\newpage

\setcounter{page}{1}
\section{Introduction}

In the Steiner forest problem, given a metric space and a set of
source-sink pairs $\{s_i, t_i\}_{i = 1}^K$, a feasible solution is a
forest such that each source-sink pair lies in the same tree in this
forest. The goal is to minimize the cost, i.e., the total length of
edges in the forest. This problem is a generalization of the Steiner
tree problem, and hence APX-hard. The constant-factor approximation
algorithms currently known for it are all based on linear programming
techniques. The first such result was an influential primal-dual
$2$-approximation due to Agrawal, Klein, and Ravi~\cite{AKR95}; this was
simplified by Goemans and Williamson~\cite{GW95} and extended to many
``constrained forest'' network design problems. Other works have since
analyzed the integrality gaps of the natural linear program, and for some
stronger LPs; see \autoref{sec:related-work}.

However, no constant-factor approximations are known based on ``purely
combinatorial'' techniques. Some natural algorithms have been proposed,
but these have defied analysis for the most part. The simplest is the
\emph{paired greedy algorithm} that repeatedly connects the
yet-unconnected $s_i$-$t_i$ pair at minimum mutual distance; this is no
better than $\Omega(\log n)$ (see Chan, Roughgarden, and
Valiant~\cite{CRV10} or \stocoption{the full version}{Appendix~\ref{sec:girth-lbd}}).
 Even greedier is the
so-called \emph{gluttonous algorithm} that connects the closest two
yet-unsatisfied terminals regardless of whether they were ``mates''.
The performance of this algorithm has been a long-standing open
question.  Our main result settles this question.
\begin{theorem}
  \label{thm:main-intro}
  The gluttonous algorithm is a constant-factor approximation for
  Steiner Forest.
\end{theorem}

We then apply this result to obtain a simple combinatorial approximation
algorithm for the \emph{two-stage stochastic version} of the Steiner
forest problem. In this problem, we are given a probability distribution
$\pi$ defined over subsets of demands. In the first stage, we can buy
some set $E_1$ of edges. Then in the second stage, the demand set is
revealed (drawn from $\pi$), and we can extend the set $E_1$ to a
feasible solution for this demand set. However, these edges now cost
$\sigma > 1$ times more than in the first stage. The goal is to minimize
the total expected cost. It suffices to specify the set $E_1$---once the
actual demands are known, we can augment using our favorite
approximation algorithm for Steiner forest. Our simple algorithm is the
following: sample $\lceil\sigma\rceil$ times from the distribution
$\pi$, and let $E_1$ be the Steiner forest constructed by (a slight
variant of) the gluttonous algorithm on union of these $\lceil
\sigma\rceil$ demand sets sampled from $\pi$.
\begin{theorem}
  \label{thm:stoc-main}
  There is a combinatorial (greedy) constant-factor approximation
  algorithm for the stochastic Steiner forest problem.
\end{theorem}
Showing that such a ``boosted sampling'' algorithm obtained a constant
factor approximation had proved elusive for several years now; the only
constant-factor approximation for stochastic Steiner forest was a
complicated primal-dual algorithm with a worse approximation
factor~\cite{GuptaK09}. Our result is based on the first cost sharing
scheme for the Steiner forest problem which is constant \gs with respect
to a constant factor approximation algorithm; see
\autoref{sec:cost-shares} for the formal definition.  Such a cost
sharing scheme can be used for designing approximation algorithms for
several stochastic network design problems for the Steiner forest
problem. In particular, we obtain the following results:
\begin{itemize}
\item For multi-stage stochastic optimization problem for Steiner
  forest, our \gs-cost sharing scheme along with the fact that it is
  also cross-monotone implies the first $O(1)^k$-approximation
  algorithm, where $k$ denotes the number of stages
  (see~\cite{GuptaPRS04} for formal definitions and the relation with
  cost sharing).
\item Consider the online stochastic problem, where given a set of
  source-sink pairs $\D$ in a metric $\M$, and a probability
  distribution $\pi$ over subsets of $\D$ (i.e., over $2^{\D}$), an
  adversary chooses a parameter $k$, and draws $k$ times independently
  from $\pi$.  The on-line algorithm, which can sample from $\pi$, needs
  to maintain a feasible solution over the set of demand pairs produced
  by the adversary at all time. The goal is to minimize the expected
  cost of the solution produced by the algorithm, where the expectation
  is over $\pi$ and random coin tosses of the algorithm.  Our cost
  sharing framework gives the first constant competitive algorithm for
  this problem, generalizing the result of Garg et al.~\cite{GargGLS08}
  which works for the special case when $\pi$ is a distribution over
  $\D$ (i.e., singleton subsets of $\D$).
\end{itemize}

\subsection{Ideas and  Techniques}
\label{sec:our-techniques}

We first describe the gluttonous algorithm. Call a terminal {\em active} if
it is not yet connected to its mate. Recall: our algorithm merges the
two active terminals that are closest in the current metric (and hence
zeroes out their distance). At any point of time, we have a collection of
{\em supernodes}, each supernode corresponding to the set of terminals
which have been merged together. A supernode is \emph{active} if it
contains at least one active terminal. Hence the algorithm can be
alternatively described thus: merge the two active supernodes that are
closest (in the current metric) into a new supernode. (A formal
  description of the algorithm appears in \S\ref{sec:algo}.)

The analysis has two conceptual steps. In the first step, we reduce
the problem to the special case when the optimal solution can be
(morally) assumed to be a single tree (formally, we reduce to the case
where the gluttonous' solution is a refinement of the optimal solution). The
proof for this part is simple: we take an optimal forest, and show that we
can connect two trees in the forest if the gluttonous algorithm connects
two terminals lying in these two trees, incurring only a factor-of-two
loss.

\begin{wrapfigure}{R}{0.3\textwidth}
  \centering
  \includegraphics[width=0.2\textwidth]{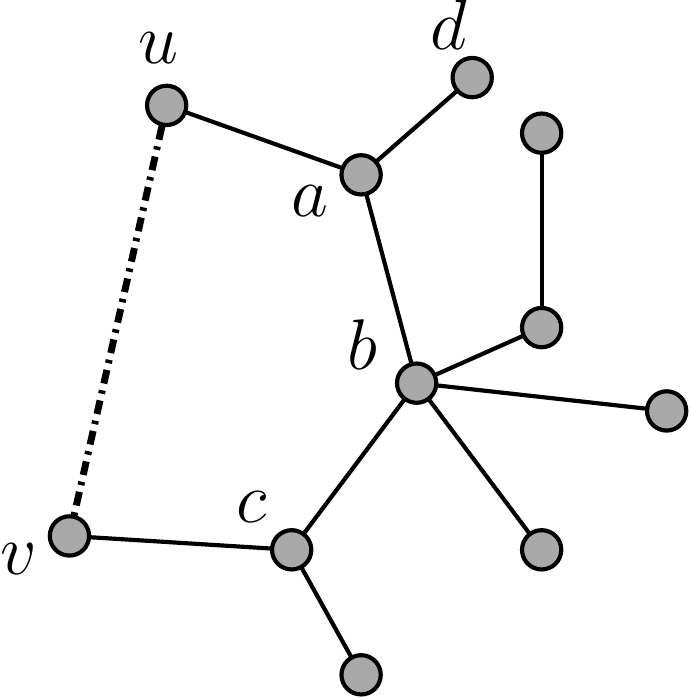}
\caption{\label{fig:tree} \footnotesize Example showing the construction of tree $T'$ from $T$, which is shown in solid lines. If we merge $u$ and $u'$,
we can remove the edge $(a,b)$ to get the tree $T'$. Assuming $a$ is not active, we can also  short-cut the degree 2 vertex $a$ in $T'$ by replacing the
edges $(s,a)$ and $(a,d)$ with the edge $(s,d)$.}

\end{wrapfigure}

The second step of the analysis starts with the tree solution $T$
promised by the first step of the analysis. As the gluttonous algorithm
proceeds, the analysis alters $T$ to maintain a candidate solution to
the current set of supernodes. E.g., if we merge two active supernodes
$u$ and $v$ to get a new supernode $uv$. We want to alter the solution $T$
on the original supernodes to get a new solution $T'$, say by removing
an edge from the (unique) $u$-$v$ path in $T$, and then short-cutting
any degree two inactive supernode in $T'$ (see Figure~\ref{fig:tree} for
an example). The hope is to argue that the distance between $u$ and
$v$---which is the cost incurred by gluttonous---is commensurate to the
cost of the edge of $T$ which gets removed during this process. This
would be easy if there were a long edge on $u$-$v$ path in the tree $T$.
The problem: this may not hold for every pair of supernodes we merge.
Despite this, our analysis shows that the bad cases cannot happen too
often, and so we can perform this charging argument in an amortized
sense.

Our analysis is flexible and extends to other variants of the gluttonous
algorithm.
A natural variant is one where, instead of merging the two
closest active supernodes, we contract the edges on a shortest path
between the two closest active supernodes.
\full{The first step of the above
analysis does not hold any more.  However, we show that it is enough to
account for the merging cost of supernodes when the active terminals in
them lie in the same tree of the optimal solution, and consequently the
arguments in the second step of the analysis are sufficient.}
Yet another
variant is a {\em timed} version of the algorithm, which is inspired by
a timed version of the primal-dual algorithm~\cite{KLSZ08}, and
is crucial for obtaining the \gs cost-shares described next.

 Loosely speaking, a cost-sharing method takes an
algorithm $\A$ and divides the cost incurred by the algorithm on an
instance among the terminals $\D$ in that instance. The ``strictness''
property ensures that if we partition $\D$ arbitrarily into $\D_1 \cup
\D_2$, and build a solution $\A(\D_1)$ on $\D_1$, then the cost-shares
of the terminals in $\D_2$ would suffice to augment the solution
$\A(\D_1)$ to one for $\D_2$ as well.

A natural candidate for $\A$ is the GW primal-dual algorithm, and the
cost-shares are equally natural: we divide up the cost of growing moats
among the active terminals in the moat. However, the example in
Figure~\ref{fig:ex} shows why this fails when $\D_2$ consists of just
the demand pair $\{s, \bar{s}\}$. When run on all the terminals, the
primal-dual algorithm stops at time~1, with all terminals getting a
cost-share of~1. On the other hand, if we run $\A$ on $\D_1$, it finds a
solution which has $N$ connected components, each connecting $s_i$ and
${\bar s}_i$ for $i=1, \ldots, N$. Then connecting $s$ and $\bar s$
costs $2N$, which is much more than their total cost share.

\begin{figure}[h]
\begin{center}
  \includegraphics[height=15mm]{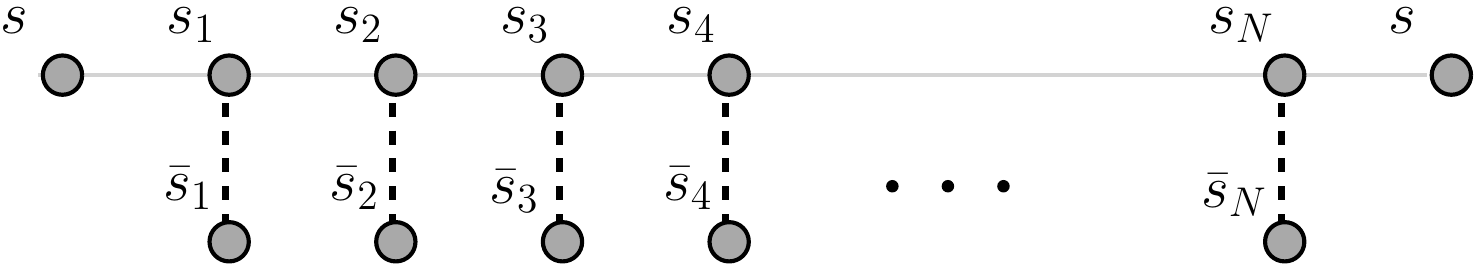}
\caption{\footnotesize Distances $d(s_i, {\bar s}_i)$ are 2 for all $i$. Further, $d(s_i, s_{i+1})=d({\bar s}_i,
{\bar s}_{i+1}) = 2$ for $i=1, \ldots, N-1$. The distances $d(s,s_1)$ and $d(s_N, {\bar s})$ are slightly larger
than 2. The dotted lines indicate the forest returned by the GW primal-dual algorithm when run on the demand set
$\{(s_i, {\bar s}_i): i=1, \ldots, N\}$}
\label{fig:ex}
\end{center}
\end{figure}

To avoid this problem,~\cite{GKPR, FKLS10} run the primal-dual algorithm
for longer than required, and give results for the case when $\D_2$
contains a single demand pair. However, the arguments become much more
involved than those in the analysis of GW algorithm~\cite{GW95}---the main reason is the presence of
``dead'' moats which cause some edges to become tight, and the cost
shares of active terminals cannot account for such edges. In our case,
the combinatorial (greedy) nature of our algorithm/analysis means  we
do not face such issues. As a result, we can obtain such strict cost sharing
methods (when $\D_2$ is a singleton set)
with much simpler analysis, albeit with worse constants than
those in~\cite{GKPR,FKLS10}). We refer to this special case of strictness property as {\em uni-strictness}.

Our analysis for the general  case where $\D_2$ contains multiple
demand pairs requires considerably more work; but note that these are
the first known strict cost shares for this case, the previous
primal-dual techniques could not handle the complexity of this general
case. Here, we want $\A$ to build as many edges as possible, and the
cost share $\chi$ to be as large as possible. Since the gluttonous
algorithm tends to build fewer edges than primal-dual (the dead moats
causing extra connections and more edges), we end up using the
primal-dual algorithm as the algorithm $\A$. However, to define the
cost-shares, we use the (timed) gluttonous algorithm in order to avoid
the issues with dead moats. The analysis then proceeds via showing a
close correspondence between the primal-dual and gluttonous
algorithms. Although this is not involved, it needs to carefully match
the two runs.

\subsubsection{Outline of Paper}

We first describe some related work in \autoref{sec:related-work}, and
give some important definitions in \autoref{sec:preliminaries}. Then we
describe the gluttonous algorithm formally in \autoref{sec:algo}, and
then analyze this algorithm in \autoref{sec:analysis}. \full{We then show that
our analysis is flexible enough to analyze several variants of the gluttonous
algorithm. We study the the timed version in
\autoref{sec:timed-version}, which gets used in subsequent sections on cost sharing.
We also consider the variant of gluttonous based on path-contraction in the appendix
(see \autoref{sec:projected}).
The cost-sharing method for the \us case is in \autoref{sec:cost-shares}, and the general case is
in \autoref{sec:group-strict}. }
\short{
We briefly describe the timed version in \autoref{sec:timed-version}, which gets
used in the subsequent section on cost sharing. The cost-sharing scheme and its
properties are described in~\autoref{sec:cost-shares}. The details of the analysis
and the simpler analysis in the \us case are deferred to the full version.
We also analyse the variant of the gluttonous which relies on contract shortest path
between supernodes in the full version. }

\subsection{Related Work}
\label{sec:related-work}

The first constant-factor approximation algorithm for the Steiner forest
problem was due to Agrawal, Klein, and Ravi~\cite{AKR95} using a
primal-dual approach; it was refined and generalized by Goemans and
Williamson~\cite{GW95} to a wider class of network design problems. The
primal-dual analysis also bounds integrality gap of the the natural LP
relaxation (based on covering cuts) by a factor of $2$. Different
approximation algorithms for Steiner forest based off the same LP, and
achieving the same factor of $2$, are obtained using the iterative
rounding technique of Jain~\cite{Jain98}, or the integer decomposition
techniques of Chekuri and Shepherd~\cite{CS08}. A stronger LP relaxation
was proposed by K\"onemann, Leonardi, and Sch\"afer~\cite{KLS05}, but it
also has an integrality gap of~$2$~\cite{KLSZ08}.

The special case of the Steiner tree problem, where all the demands
share a common (source) terminal, has been well-studied in the network
design community. There is a simple 2-approximation algorithm for this
problem: iteratively find the closest terminal to the source vertex, and
merge these two terminals. There have been several changes to this
simple greedy algorithm leading to improved approximation ratios~(see
e.g.~\cite{RobinsZ05}).  Byrka et al.~\cite{ByrkaGRS13} improved these
results to a $\ln 4 + \varepsilon \approx 1.46$-approximation algorithm,
which is based on rounding a stronger LP relaxation for this problem.

The stochastic Steiner tree/forest problem was defined by Immorlica,
Karger, Minkoff, and Mirrokni~\cite{IKMM04}, and further studied
by~\cite{GuptaPRS04}, who proposed the boosted-sampling framework of
algorithms. The analysis of these algorithms is via ``strict'' cost
sharing methods, which were studied by~\cite{GKPR,FKLS10}.  A
constant-factor approximation algorithm (with a large constant) was
given for stochastic Steiner forest by~\cite{GuptaK09} based on
primal-dual techniques; it is much more complicated than the algorithm
and analysis based on the greedy techniques in this paper.

\subsection{Preliminaries}
\label{sec:preliminaries}

Let $\M = (V,d)$ be a metric space on $n$ points; assume all distances
are either $0$ or at least~$1$. Let the \emph{demands} $\D \subseteq
\binom{V}{2}$ be a collection of source-sink pairs that need to be
connected. By splitting vertices, we may assume that the pairs in $\D$
are disjoint. A node is a \emph{terminal} if it belongs to some pair in
$\D$. Let $K$ denote the number of terminals pairs, and hence there are
 $2K$ terminals. For a terminal $u$, let $\bar{u}$ be the unique
vertex such that $\{u, \bar{u}\} \in \D$; we call $\bar{u}$ the
\emph{mate} of $u$.

For a \stf instance $\I = (\M, \D)$, a solution $\F$ to the instance $\I$
is a forest such that each pair $\{u, \bar{u}\} \in \D$ is contained
within the vertex set $V(T)$ for some tree $T \in \F$. For a tree $T =
(V, E_T)$, let $\cost(T) := \sum_{e \in E_T} d(e)$ be the sum of lengths
of edges in $T$. Let $\cost(\F) := \sum_{T \in \F} \cost(T)$ be the cost
of the forest $\F$. Our goal is to find  a solution of minimum cost.

\section{The Gluttonous Algorithm}
\label{sec:algo}

To describe the gluttonous algorithm, we need some definitions. Given a
\stf instance $\I = (\M, \D)$, a \emph{supernode} is a subset of
terminals. A \emph{clustering} $\C = \{S_1, S_2, \ldots, S_q\}$ is a
partition of the terminal set into supernodes. The \emph{trivial
  clustering} places each terminal in its own singleton supernode.  Our
algorithm maintains a clustering at all points in time.  Given a
clustering, a terminal $u$ is \emph{active} if it belongs to a supernode
$S$ that does not contain its mate $\bar{u}$. A supernode $S$ is
\emph{active} if it contains some active terminal. In the trivial
clustering, all the terminals and supernodes are active.

Given a clustering $\C = (S_1, S_2, \ldots, S_q)$, define a new metric
$\M/\C$ called the \emph{$\C$-puncturing} of metric $\M$. To get this,
take a complete graph on $V$; for an edge $\{u,v\}$, set its length to
be $d(u,v)$ if $u,v$ lie in different supernodes in $\C$, and to zero
if $u,v$ lie in the same supernode in $\C$.  Call this graph $G_\C$, and
defined the $\C$-punctured distance to be the shortest-path distance in
this graph, denoted by $d_{\M/\C}(\cdot, \cdot)$. One can think of this
as modifying the metric $\M$ by collapsing the terminals in
each of the supernodes in $\C$ to
a single node.  Given clustering $\C$
and two supernodes $S_1$ and $S_2$, the distance between them is
naturally defined as
\[ d_{\M/\C}(S_1, S_2) = \min_{u \in S_1, v \in S_2} d_{\M/\C}(u, v). \]
The gluttonous algorithm is as follows:
\begin{quote}
  Start  with $\C$ being the trivial clustering, and $E'$ being the
  empty set. While there exist active supernodes in $\C$, do the
  following:
  \begin{OneLiners}
  \item[(i)] Find active supernodes $S_1, S_2$ in $\C$ with minimum
    $\C$-punctured distance. (Break ties arbitrarily but consistently,
    say choosing the lexicographically smallest pair.)
  \item[(ii)] Update the clustering to
    \[ \C \gets (\C \setminus \{S_1, S_2\}) \cup \{ S_1 \cup S_2 \}, \]
  \item[(iii)] Add to $E'$ the edges corresponding to the inter-supernode
    edges on the shortest path between $S_1, S_2$ in the graph $G_\C$.
  \end{OneLiners}
  Finally, output a maximal acyclic subgraph $F$ of $E'$.
\end{quote}
Above, we say we \emph{merge} $S_1, S_2$ to get the new supernode $S_1 \cup
S_2$. The \emph{merging distance} for the merge of $S_1, S_2$ is the
$\C$-punctured distance $d_{\M/\C}(S_1, S_2)$, where  $\C$ is the clustering
just before the merge. Since each active supernode contains an active
terminal, if $u \in S_1$ and $v \in S_2$ are both active, then
when we talk about merging $u, v$, we mean merging $S_1,S_2$.

Note that the length of the edges added in step~(iii) is equal to
$d_{\M/\C}(S_1, S_2)$. The algorithm maintains the following invariant:
if $S$ is a supernode, then the terminals in $S$ lie in the same
connected component of $F$.\footnote{The converse is not necessarily
  true: if we connect $S_1$ and $S_2$ by buying edges connecting them
  both to some inactive supernode $S_3$, then $F$ has a tree connecting
  all three, but the clustering has $S_1 \cup S_2$ separate from
  $S_3$. Indeed, inactive supernodes never get merged again, whereas
  inactive trees may.} The algorithm terminates when there are no more
active terminals, so each terminal shares a supernode with its mate, and
hence the final forest $F$ connects all demand pairs. Since the edges
added to $E'$ have total length at most the sum of the merging
distances, and we output a maximal sub-forest of $E'$, we get:

\begin{fact}
  \label{fct:can-make-forest}
  The cost of the \stf solution output is  at most the sum of all
  the merging distances.
\end{fact}

We emphasize that the edges added in Step~(iii) are often overkill: the
metric $\M/E'$ (where the edges in $E'$ have been contracted) has no
greater distances than the metric $\M/\C$ that we focus on. The
advantage of the latter over the former is that distances in $\M/\C$ are
well-controlled (and distances between active terminals only increase
over time), whereas those in $\M/E'$ change drastically over time (with
distances between active terminals changing unpredictably).

\begin{wrapfigure}{R}{0.5\textwidth}
  \centering
  \includegraphics[height=15mm]{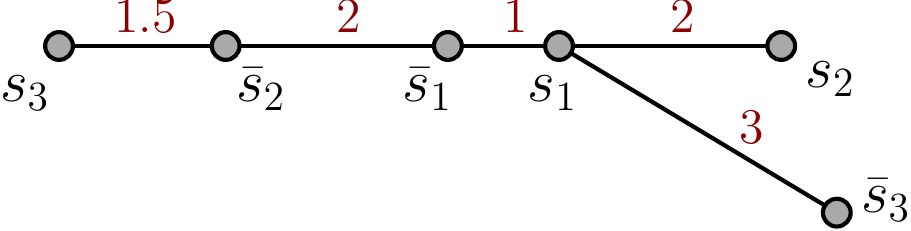}
  \caption{\footnotesize Figure for the gluttonous algorithm.}
  \label{fig:glut}
\end{wrapfigure}

Consider the example in Figure~\ref{fig:glut}, where the distances for
missing edges are inferred by computing shortest-path
distances. 

Here, we first merge $\{s_1, {\bar s}_1\}$ to form a supernode, say $A$,
which is inactive. Next we merge $s_3$ and ${\bar s}_2$ to form another
supernode, say $B$. The active supernodes are $B, s_2$, and $\bar{s}_3$,
so we next merge $s_2$ with $B$ to form supernode $C$, and finally merge
${\bar s}_3$ with $C$. When the algorithm ends, there are two (inactive)
supernodes corresponding to the sets $\{s_1, {\bar s}_1\}$ and $\{s_2,
s_3, {\bar s}_2, {\bar s}_3 \}$. However, the forest produced will have
only a single tree, which consists of the set of edges drawn in the figure.

\section{The Analysis for Gluttonous}
\label{sec:analysis}

We analyze the algorithm in two steps. One conceptual problem is in
controlling what happens when gluttonous connects two nodes in different
trees of the optimal forest. To handle this, we show in
\autoref{sec:near-optimal} how to preprocess the optimal forest $\Fs$ to
get a near-optimal forest $\Fss$ such that the final clustering of the
gluttonous algorithm is a refinement of this near-optimal forest. (I.e.,
if $u$ and $v$ are in the same supernode in the gluttonous clustering,
then they lie in the same tree in $\Fss$.)  This makes it easier to then
account for the total merging distance, which we do in
\autoref{sec:charging}. \short{A crucial observation, which makes the
  analysis much cleaner, is the following.}

\full{
\subsection{Monotonicity Properties}
\label{sec:mono}

To begin, some simple claims about monotonicity. The first one is by
definition.

\begin{fact}[Distance Functions are Monotone]
  \label{fct:dist-mono}
  Let the clustering $\C'$ correspond to a later time than the clustering
  $\C$. Then $\C$ is a refinement of $\C'$. Moreover, $d_{\M/\C'}(u, v)
  \leq d_{\M/\C}(u,v)$ for all $u, v \in V$.
\end{fact}

\begin{claim}
  \label{clm:close-mono}
  Consider clustering $\C$ and let any two active supernodes $S,T$ be
  merged, resulting in clustering $\C'$. Then for any active $U \in \C$
  that is not $S$ or $T$, the distance to its closest active supernode
  does not decrease. Also, if $S \cup T$ is active in $\C'$ then the
  distance to its closest supernode in $\C'$ is at least as large as the
  minimum of $S$ and $T$'s distances to their closest supernodes in $\C$.
\end{claim}

\begin{proof}
  First observe that when two active supernodes are merged, they may
  stay active or become inactive. An inactive supernode never merges with
  any other supernode, and hence, it cannot become active later.

  For active supernode $U \neq S, T$, suppose its closest supernode in
  $\C$ was $W$ at $\C$-punctured distance $L$, and in $\C'$ it is $W'$
  at $\C'$-punctured distance $L'$. If $L' < L$, there must now be a path
  through supernode $S \cup T$ that is of length $L'$. But this means
  the $\C$-punctured distance of $U$ from either $S$ or $T$ was at most
  $L'$, and both were active in $\C$---a contradiction. This proves the
  first part of the claim.

  Now, suppose supernode $U := S \cup T$ is active in $\C'$. Observe
  that for any other supernode $W \in \C'$, the punctured distance
  $d_{\M/\C'}(U, W) = \min\{ d_{\M/\C}(S, W), d_{\M/\C}(T, W)\}$. This
  proves the second part of the claim.
\end{proof}
}

\begin{claim}[Gluttonous Merging Distances are Monotone]
  \label{clm:merge-mono}
  If $S,T$ are merged before $S', T'$ in gluttonous, then the merging
  distance for $S,T$ is no greater than the merging distance for $S',
  T'$.
\end{claim}

\full{
\begin{proof}
  Gluttonous merges two active supernodes with the smallest current
  distance.  By Claim~\ref{clm:close-mono} distances between the
  remaining active supernodes do not decrease. This proves this claim.
\end{proof}
}

\subsection{A Near-Optimal Solution with Good Properties}
\label{sec:near-optimal}

Since gluttonous is deterministic and we break ties consistently, given
an instance \stf instance $\I = (\M, \D)$ there is a unique final
clustering $\C^f$  produced by the algorithm.

\begin{definition}[Faithful]
  \label{def:faithful}
  A forest $\F$ is \emph{faithful} to a clustering $\C$ if each
  supernode $S \in \C$ is contained within a single tree in~$\F$. (I.e.,
  for all $S \in \C$, there exists $T \in \F$ such that $S \sse V(T)$.)
\end{definition}
Note that every forest is faithful to the trivial clustering consisting
of singletons.
\full{
\begin{definition}[Width]
  \label{def:width}
  For a forest $\F$ that is a solution to instance $\I$, and for any
  tree $T \in \F$, let $\width(T)$ denote the largest tree distance
  between any pair connected by $T$. Let the width of forest $\F$ be the
  sum of the widths of the trees in $\F$. I.e.,
  \begin{align}
    \width(T) &:= \max \{ d_T(u, \bar{u}) \mid \{u, \bar{u}\} \in \D,
    \{u, \bar{u}\} \sse V(T) \}, \\
    \width(\F) &:= \textstyle \sum_{T \in F} \width(T),
  \end{align}
  where $d_T$ refers to the tree metric induced by $T$.
\end{definition}

We now show there exist near-optimal solutions
which are faithful to gluttonous' final clustering.
}
\short{The following theorem, whose proof is deferred to the full version,
shows that there exist near-optimal solutions
which are faithful to gluttonous' final clustering.}
\begin{theorem}[Low-Cost and Faithful]
  \label{thm:one-tree}
  Let $\Fs = \{T_1^\star, T_2^\star, \ldots, T_p^\star\}$ be an optimal
  solution to the \stf instance $\I = (\M, \D)$. There exists another
  solution $\Fss$ for instance $\I$ such that
  \begin{OneLiners}
  \item[(a)] $\cost(\Fss) \full{ \leq \cost(\Fs) + \width(\Fs)}  \leq
    2\cost(\Fs)$, and
  \item[(b)] $\Fss$ is faithful to the final clustering $\C^f$ produced
    by the gluttonous algorithm. 
  \end{OneLiners}
\end{theorem}

\full{
\begin{proof}
  Start with $\Fss = \Fs$ which clearly satisfies the first (cost)
  guarantee but perhaps not the second (faithfulness) one. To fix this,
  run the gluttonous algorithm on $\I$, and whenever it connects two
  terminals $(u,v)$ that violate the condition~(b), connect up some
  trees in the current $\Fss$ to prevent this violation. In particular,
  we show how to do this while maintaining two invariants:
  \begin{itemize}
  \item[(A)] The
    cost of  edges in $\Fss \setminus \Fs$ is at most $
    \width(\Fs) - \width(\Fss)$, and
  \item[(B)] at any point in time, the forest $\Fss$ is faithful to the
    current clustering $\C$ (during the run of the gluttonous algorithm).
  \end{itemize}
  At the beginning, the clustering $\C$ is the trivial clustering
  consisting of singleton sets containing terminals, and $\Fss =
  \Fs$; both invariants~(A) and~(B) are vacuously true.

  Now consider some step of gluttonous which starts with the clustering
  $\C$ and connects two active supernodes $S$ and $S'$ which are closest
  to each other to get the clustering $\C'$. By the invariant~(B), we
  know all terminals in supernode $S$ lie within the same tree in
  $\Fss$, and the same for terminals in $S'$. Let $u \in S$ and $v \in
  S'$ be some active terminals within these supernodes; hence $\bar u
  \not\in S$ and $\bar v \not\in S'$.
  Two cases arise:

  \begin{itemize}
  \item \textbf{Case I:} $u$ and $v$ belong to the same tree in $\Fss$:
    Clearly, $\Fss$ satisfies the invariant~(B) with respect to $\C'$ as
    well. Hence, we keep $\Fss$ unchanged and it satisfies invariant~(A)
    trivially.

  \item \textbf{Case II:} $u$ and $v$ belong to different trees
    $T_1^{\star\star}, T_2^{\star\star} \in \Fss$: Suppose the shortest
    path between $u$ and $v$ in $\M/\C$ is
    \[ P = \{ u = x_0, x_0', x_1, x_1', x_2, x_2', x_3, \ldots,
    x_{k-1}', x_k, x_k' = v\} \] such that each $x_i, x_i'$ belong to
    the same supernode $S_i$ in $\C$ (see e.g., figure~\ref{fig:single}).  By the
    greedy behavior of gluttonous, $\cost(P)$ is at most the cost to
    connect $u$ to $\bar{u}$, or to connect $v$ to $\bar{v}$ in
    $\M/\C$. In fact, we can bound these costs by the cost of the edges
    between $u, \bar{u}$ in $T_1^{\star\star}$, etc.  Hence,
    \begin{eqnarray}
      \label{eq:cost}
      \cost(P) \leq \min \big\{ d_{T_1^{\star\star}}(u, \bar{u}),
      d_{T_2^{\star\star}}(v, \bar{v}) \big\} \leq \min \big\{
      \width(T_1^{\star\star}), \width(T_2^{\star\star}) \big\} .
    \end{eqnarray}
    Since each of the supernodes $S_i$ is contained within some tree in
    $\Fss$ (by invariant~(B) applied to clustering $\C$), we need only
    add (a subset of edges from) the path $P$ to the forest $\Fss$ in
    order to merge $T_1^{\star\star}$ and $T_2^{\star\star}$ (and
    perhaps other trees in $\Fss$) into one single tree---thus ensuring
    invariant~(B) for the new clustering $\C'$.

    \begin{figure}[h]
      \begin{center}
        \includegraphics[height=30mm]{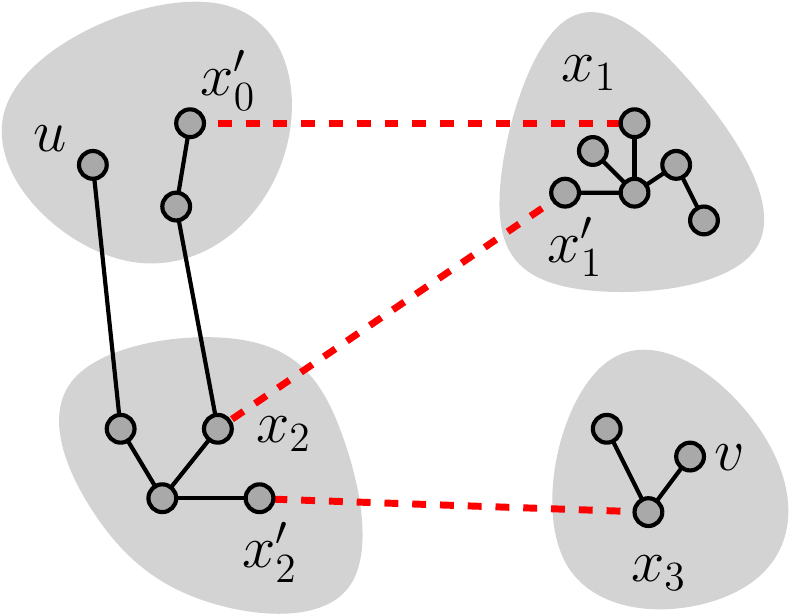}
        \caption{\footnotesize Case II of the proof of
          Theorem~\ref{thm:one-tree}. The grey blobs are supernodes in
          $\C$, the solid lines denote the forest $\Fss$. The dotted
          lines are the path $P$. Observe we do not need to add the
          second edge of $P$ as it will create a cycle. }
        \label{fig:single}
      \end{center}
    \end{figure}

    How does the width of the trees in $\Fss$ change? Each tree that we
    merge is inactive (since it is a coarsening of the original solution
    $\Fs$). Connecting up $T_1^{\star\star}, \ldots,
    T_k^{\star\star}$ causes the width of the resulting tree to be $\max
    \{ \width(T_1^{\star\star}), \cdots,
    \width(T_k^{\star\star})\}$. The decrease in $\width(\Fss)$ due to
    the merge is at least $\min \{ \width(T_1^{\star\star}),
    \width(T_2^{\star\star})\}$, which ensures invariant~(A) (using
    inequality~(\ref{eq:cost})).
  \end{itemize}

  Hence, at the end of the run of gluttonous, both invariants hold.
  Since the initial potential is $\width(\Fs) \leq \cost(\Fs)$, and the
  final potential is non-negative, the total cost of edges in $\Fss
  \setminus \Fs$ is at most $\cost(\Fs)$. This completes the proof.
\end{proof}
}


\subsection{Charging to this Near-optimal Solution}
\label{sec:charging}

\renewcommand{\Kss}{\ensuremath{K^{\star\star}}}
\renewcommand{\del}{\mathsf{del}}

Let $\Fs = \{T^\star_1, \ldots, T^\star_p\}$ be a solution to the \stf instance.
The main result of this section is:

\begin{theorem}
  \label{thm:one-tree-is-cheap}
  If the forest $\Fs$ is faithful to the final clustering $\C^f$ of the
  gluttonous algorithm, then the cost of the gluttonous algorithm is
  $O(1) \cdot \cost(\Fs)$.
\end{theorem}

Since by Theorem~\ref{thm:one-tree} there is a forest $\Fs$ with cost
at most twice the optimum that is faithful to gluttonous' final
clustering $\C^f$, applying Theorem~\ref{thm:one-tree-is-cheap} to this
forest proves Theorem~\ref{thm:main-intro}.

We now prove Theorem~\ref{thm:one-tree-is-cheap}. At a high level, the
proof proceeds thus: we consider the run of the gluttonous algorithm,
and maintain for each iteration $t$ a ``candidate'' forest $\F_t$ that
is a solution to the remaining instance.  We show that in an amortized
sense, at each step the cost of forest $\F_t$ decreases by an amount
which is a constant fraction of the cost incurred by gluttonous. Since
the starting cost of this forest is at most a constant times the optimal cost,
so is the total merging
cost of the gluttonous, proving the result.

For Steiner forest instance $\I$, assume that $\C^f$ is gluttonous'
final clustering, and $\Fs$ is faithful to $\C^f$.
Let $\clus{t}$ be the gluttonous clustering at the beginning of the
iteration $t$, with $\activecl{t}$ being the active supernodes. It will
be useful to view this clustering as giving us an induced Steiner forest
instance $\I_t$ on the metric whose points are the supernodes in
$\clus{t}$ and where distances are given by the punctured metric
$d_{\M/\clus{t}}$, where the terminals in the instance $\I_t$ are
supernodes in $\activecl{t}$, and where active supernodes $\{S_1, S_2\}$
are mates if there is a pair $\{u, \bar{u}\}$ such that $u \in S_1$ and
$\bar{u} \in S_2$. (Supernodes no longer have unique mates, but this
property was only used for convenience in
Theorem~\ref{thm:one-tree}). For any iteration $t$, the subsequent run
of gluttonous is just a function of this induced instance
$\I_t$. Indeed, given the instance $\I_t$, gluttonous outputs a final
clustering which is same as $\C^f$ except the inactive supernodes in
$\clus{t}$ are absent.  I.e., the inactive supernodes in $\clus{t}$ will
not play a role, but all the active supernodes will continue to combine
in the same way in $\I_t$ as in $\I$. We now inductively maintain a
forest $\fcl{t}$ such that
\begin{OneLiners}
\item[(I1)] $\fcl{t}$ is a feasible solution to this \stf instance $\I_t$,
  and

\item[(I2)] $\fcl{t}$ maintains the connectivity structure of $\Fs$,
  i.e., if $u$ and $v$ are two active terminals which are in the same tree in
  $\Fs$, then the supernodes containing $u$ and $v$ lie in the same
  tree in $\fcl{t}$.
\end{OneLiners}
And we will charge the
cost of gluttonous to reductions in the cost of this forest $\fcl{t}$.

\paragraph{The ``candidate'' forest $\fcl{t}$.}
The initial clustering $\clus{1}$ is the trivial clustering consisting
of singleton terminals; we set $\fcl{1}$ to $\Fs$. Since $\I_1$ is the
original instance, $\fcl{1}$ is feasible for it; invariant~(I2) is
satisfied trivially.

For an iteration $t$, let $E(\fcl{t})$ denote the edges in
$\fcl{t}$. Note that an edge $e \in E(\fcl{t})$ between two supernodes
$S_1, S_2 \in \clus{t}$ corresponds to an edge between two terminals
$u,v$ in the original metric $\M$, where $u \in S_1, v \in S_2$. Define
$\length(e)$ as $d_\M(u,v)$, the length of the edge $e$ in the original
metric.  Note that the length of $e$ in the metric $\M/\clus{t}$ may be
smaller than $\length(e)$. For every edge $e \in \fcl{t}$, we shall also
maintain a {\em potential} of $e$, denoted $\pl(e)$.  Initially, for
$t=1$, the potential $\pl(e)=\length(e)$ for all $e \in E(\fcl{1})$.
During the course of the algorithm, the potential $\pl(e) \geq
\length(e)$; we describe the rule for maintaining potentials
below. Intuitively, an edge $e \in \fcl{t}$ would have been obtained by
{\em short-cutting} several edges of $\Fs$, and $\pl(e)$ is equal to
the total length of these edges.

Suppose we have a clustering $\clus{t-1}$ and a forest $\fcl{t-1}$ which
satisfies invariants (I1) and (I2). If we now merge two supernodes $S_1,
S_2 \in \clus{t-1}$ to get clustering $\clus{t}$, we have to update the
forest $\fcl{t-1}$ to get to $\fcl{t}$ using procedure \updateforest
given in Figure~\ref{fig:update}. The main idea is simple: when we merge
the nodes corresponding to $S_1$ and $S_2$ in $\fcl{t-1}$ into a single
node, this creates a cycle. Removing any edge from the cycle maintains
the invariants, and reduces the cost of the new forest: we remove the
edge with the highest potential from the cycle. We further reduce the
cost by getting rid of Steiner vertices, which correspond to inactive
supernodes in $\fcl{t}$ with degree~2. More formally, given two edges
$e'=\{u',v\}, e''=\{u'',v\}$ with a common end-point $v$, the operation
{\em short-cut} on $e',e''$ replaces them by a single edge $\{u',u''\}$.
Whenever we see a Steiner vertex of degree~2 in $\fcl{t}$, we shortcut
the two incident edges.

\begin{figure}[h]
  \centering
  \begin{boxedminipage}{0.9\textwidth}
    {\bf Algorithm \updateforest($\clus{t-1}, S_1, S_2$) :} \medskip\\
    \sp \sp \sp 1. Let $T$ be the tree in $\fcl{t-1}$ containing the terminals in $S_1$ and $S_2$. \\
    \sp \sp \sp 2. Merge $S_1$ and $S_2$ to a single node $S$ in the tree $T$. \\
    \sp \sp \sp 3. If the new supernode $S$ becomes inactive, and has degree 2 in the tree $T$, then  \\
    \sp \sp \sp \sp \sp \sp \sp short-cut the two edges incident to $S$. \\
    \sp \sp \sp 4. Let $C$ denote the unique cycle formed in the tree $T$. \\
    \sp \sp \sp 5. Delete  the edge in the cycle $C$ which has the highest potential. \\
    \sp \sp \sp 6. While there is an inactive supernode in $T$ which is a degree-2 vertex, \\
    \sp \sp \sp \sp \sp \sp \sp short-cut the two incident edges to this vertex.
  \end{boxedminipage}
  \label{fig:update}
  \caption{\footnotesize The procedure for updating $\fcl{t-1}$ to $\fcl{t}$.}
\end{figure}

Some more comments about the procedure~\updateforest. In Step 1, the
existence of the tree $T$ follows from the invariant property (I2) and
the faithfulness of $\Fs$ to $\C^f$. Since the terminals in $S_1 \cup
S_2$ are in the same tree in $\Fs$, the invariant means they belong to
the same tree in $\fcl{t-1}$, and the construction ensures they remain
in the same tree in $\fcl{t}$. When we short-cut edges $e', e''$ to get
a new edge $e$, we define the potential of the new edge $e$ to be
$\pl(e) := \pl(e')+\pl(e'')$.  It is also easy to check that $\fcl{t}$
is a feasible solution to the instance $\I_t$. Indeed, the only
difference between $\I_{t-1}$ and $\I_t$ is the replacement of $S_1, S_2$
by $S$.  If  $S$ becomes inactive, there is nothing to prove. If
$S$ remains active, then the tree containing $S$ must will also have
also have the supernodes which were paired with $S_1$ and $S_2$
in the instance $\I_{t-1}$. It is also easy to check that the invariant
property (I2) continues to hold.
The following claim\short{, whose proof is deferred to the full version,}
 proves some more crucial
properties of the forest $\fcl{t}$.
\begin{claim}
  \label{cl:forest-prop}
  For all iterations $t$, the Steiner nodes in $\fcl{t}$ have degree at
  least 3. Therefore, there are at most 2 iterations of the while loop
  in Step 6 of the \updateforest algorithm.
\end{claim}
\full{
\begin{proof}
  We prove the first statement of the lemma by induction on $t$. For
  $t=1$, it holds by construction: we can assume that $\Fs$ has no
  Steiner vertex of degree at most 2: any leaf Steiner node can be
  deleted, and a degree 2 can be removed by short-cutting the incident
  edges. Suppose this property is true for $\fcl{t-1}$. We merge $S'$
  and $S''$, and if the new supernode $S$ becomes an inactive supernode,
  then its degree will be at least 2 (both $S'$ and $S''$ must have had
  degree at least 1).  If the degree is equal to 2, we remove this
  vertex in Step 3.

  When we remove an edge in Step 5, the two end-points could have been
  Steiner vertices. By the induction hypothesis, their degree will be at
  least 2 (after the edge removal). If their degree is 2, we will again
  remove them by short-cutting edges. Note that this will not affect the
  degree of other nodes in the forest.  This also shows that Step 6 will
  be carried out at most twice.
\end{proof}
}

Here's the plan for rest of the analysis. Let's fix a tree $\Ts$ of
$\Fs$, and account for only those merging costs which merge two
supernodes with terminals in $\Ts$. (Summing over all trees in $\Fs$
and using the faithfulness of $\Fs$ to $\C^f$ will ensure all merging
costs are accounted for.) Since $\fcl{t}$ is obtained by repeatedly
contracting nodes and removing unnecessary edges, in each iteration $t$
there is a unique tree $\tcl{t}$ in the forest $\fcl{t}$ corresponding
to the tree $\Ts$, namely the tree containing the active supernodes
with terminals belonging to $\Ts$. Call an iteration of the gluttonous
algorithm a {\em relevant iteration} (with respect to $\Ts$) if
gluttonous merges two supernodes from the tree $\tcl{t}$ in this
iteration. For brevity, we drop the phrase ``\emph{w.r.t.\ $\Ts$}'' in
the sequel.

Next we show that the total potential of the edges does not change over
time. Let $\del(t)$ denote the set of edges 
which are deleted (from a cycle in Step 5) during the (relevant)
iterations among $1, \ldots, t-1$. (Observe that $\del(t)$ does not include
  edges that are short-cut.)
\begin{lemma}
  \label{lem:forest-inv}
  For iteration $t$, the sum of potentials of edges $\del(t)$ and
  $E(\tcl{t})$ equals $\cost(\Ts)$. Further, $\pl(e) \geq
  \length(e)$ for all edges $e \in E(\tcl{t}) \cup \del(t)$.
\end{lemma}
\full{
\begin{proof}
  By induction on $t$. The base case $t=1$ follows by construction. For
  the IH, assume the statement holds for $t-1$. Assume that $t$ is a
  relevant iteration (else $\tcl{t} = \tcl{t-1}$): if we remove edge $e$
  from $\tcl{t-1}$ during Step 5, we do not change $\pl(e)$.  If we
  short-cut two edges $e', e''$ to an edge $e$,
  $\pl(e)=\pl(e')+\pl(e'')$. Therefore the total potential of the edges
  in the tree plus that of the edges in $\del(t)$ does not change.
  Further, $\length(e) \leq \length(e') + \length(e'') \leq
  \pl(e')+\pl(e'') = \pl(e)$.
\end{proof}
}

Eventually $\tcl{t}$ has no active supernodes (for large $t$) and hence
all its edges are deleted. Hence if $\del(\infty)$ denotes the edges
deleted during all the relevant iterations in gluttonous,
Lemma~\ref{lem:forest-inv} implies $\sum_{e \in \del(\infty)} \pl(e) =
\cost(\Ts)$.
Let $\Delta_t$ denote the merging cost of some relevant iteration $t$:
we now show how to charge this cost to the potential of some deleted
edge in $\del(\infty)$. Formally, let $N_t$ denote the number of active
supernodes in $\tcl{t}$, at the \emph{beginning} of iteration $t$.
\begin{theorem}
  \label{thm:charge}
  If $t_0$ is relevant, there are at least $N_{t_0}/8$ edges in
  $\del(\infty)$ of potential at least $\Delta_{t_0}/6$.
\end{theorem}

We defer the proof of Theorem~\ref{thm:charge} for the moment, and
instead show how to use this to charge the merging costs and to prove
Theorem~\ref{thm:one-tree-is-cheap}, which in turn gives the main
theorem of the paper.

\begin{proofof}{Theorem~\ref{thm:one-tree-is-cheap}}
  Let $I^r$ denote the index set of all relevant iterations during the
  run of gluttonous. We now define a mapping $g$ from $I^r$ to
  $\del(\infty)$ such that: (i) for any edge $e \in \del(\infty)$, the
  pre-image $g^{-1}(e)$ has cardinality at most 8, and (ii) the
  potential $\pl(g(t)) \geq \Delta_t/6$ for all $t \in I^r$. To get
  this, consider a bipartite graph on vertices $I^r \cup \del(\infty)$
  where a iteration $t \in I^r$ is connected to all edges $e\in
  \del(\infty)$ for which $\pl(e) \geq \Delta_t/6$. Theorem~\ref{thm:charge}
  shows this graph satisfies a Hall-type condition for such a mapping to
  exist; in fact a greedy strategy can be used to construct the mapping (there
  can be at most $N_t$ relevant iterations after iteration $t$ because each
  relevant iteration reduces the number of active supernodes by at least one).

  Thus, the total merging cost of gluttonous during relevant iterations
  is at most
  \[ \short{\ts} \sum_{t \in I^r} \Delta_t = \sum_{e \in \del(\infty)} \sum_{t \in
    g^{-1}(e)} \Delta_t \leq 48 \sum_{e \in \del(\infty)} \pl(e) = 48 \,
  \cost(\Ts), \] where the last equality follows from
  Lemma~\ref{lem:forest-inv}. By the faithfulness property, each
  iteration of gluttonous is relevant with respect to one of the trees
  in $\Fs$, so summing the above expression over all trees gives the
  total  merging cost to be at
  most $48 \, \cost(\Fs)$.
\end{proofof}

Combining Theorem~\ref{thm:one-tree-is-cheap} with
Theorem~\ref{thm:one-tree} gives an approximation factor of $96$ for the
gluttonous algorithm. While we have not optimized the constants, but it is
unlikely that our ideas will lead to constants in the single
digits. Obtaining, for instance, a proof that the gluttonous algorithm
is a $2$-approximation (or some such small constant) remains a
fascinating open problem.

\subsubsection{Proof of Theorem~\ref{thm:charge}}

In order to prove Theorem~\ref{thm:charge}, we need to understand the
structure of the trees $\tcl{t}$ for $t \geq t_0$ in more detail.
Let $\del_0([t_0 \ldots t))$ denote the edges deleted during the
relevant iterations in $t_0, \ldots, t-1$, i.e., $\del([t_0\ldots t)) :=
\del(t) \setminus \del(t_0)$.  Observe that each edge of $\tcl{t}$ is
either in $\tcl{t_0}$ or is obtained by short-cutting some set of edges
of $\tcl{t_0}$. Hence we maintain a partition $\E(t)$ of the edge set
$E(\tcl{t_0})$, such that there is a correspondence between edges $e \in
\tcl{t} \cup \del([t_0 \ldots t))$ and sets $D_t(e) \in \E(t)$, such that
$D_t(e)$ is the set of edges in $\tcl{t_0}$ which have been short-cut to
form $e$.

For each set $D_t(e)$, let $\head(D_t(e))$ be the edge $e' \in D_t(e)$
with greatest length. If edge $e$ is removed from $\tcl{t}$ in some
relevant iteration $t$, we have $e \in \del([t_0 \ldots t'))$ for all
$t' > t$, and the set $D_{t'}(e) = D_t(e)$ for all future partitions
$\E(t')$.

\begin{lemma}
  \label{lem:long}
  There are at least $N_{t_0}/2$ edges of length at
  least $\Delta_{t_0}/6$ in tree $\tcl{t_0}$.
\end{lemma}
\begin{proof}
  Call an edge \emph{long} if its length is at least $\Delta_{t_0}/6$,
  and let $\ell$ denote the number of long edges in the tree
  $\tcl{t_0}$. Deleting these edges from $\tcl{t_0}$ gives $\ell+1$
  subtrees $C_1, C_2, \ldots, C_{l+1}$. Let $C_i$ have $n_i$ active
  supernodes and $e_i$ edges. For each tree $C_i$ where $n_i \geq 2$,
  take an Eulerian tour $X_i$ and divide it into $n_i$ disjoint segments by
  breaking the tour at the active supernodes. Each edge appears in two
  such segments, and each segment has at least six edges (since the
  distance between active supernodes is at least $\Delta_{t_0}$ and none
  of the edges are long), so $e_i \geq 3n_i$ when $n_i \geq 2$.  This
  means the total number of edges in $\tcl{t_0}$ is at least three times
  the number of ``social'' supernodes (supernodes that do not lie in a
  component $C_i$ with $n_i = 1$, in which they are the only supernode),
  \emph{plus} those $\ell$ long edges that were deleted.

  And how many such social supernodes are there? If $\ell \geq N_{t_0} + 1$,
  there may be none, but then we clearly have at least $N_{t_0}/2$ long
  edges. Else at least $N_{t_0} - \ell$ supernodes are social, so
  $\tcl{t_0}$ has at least $3(N_{t_0} - \ell) + \ell$ edges. Finally,
  since every Steiner vertex in $\tcl{t_0}$ has degree at least $3$, the
  number of edges is less than $2N_{t_0}$. Putting these together gives
  $3\,N_{t_0} - 2\ell \leq 2\,N_{t_0}$ or $\ell \ge N_{t_0}/2$.
\end{proof}

Let $L_0$ be the set of long edges in $\tcl{t_0}$, and $\E(\infty)$ be
the partition at the end of the process. Two cases arise:
\begin{itemize}
\item At least $N_{t_0}/8$ edges in $L_0$ are $\head(D_\infty(e))$ for
  some set $D_\infty(e) \in \E(\infty)$. Since each set in $\E(\infty)$
  has only one head, there are $N_{t_0}/8$ such sets. In any such set
  $D_\infty(e)$, $\pl(e) \geq \length(\head(D_\infty(e))) \geq
  \Delta_{t_0}/6$. Moreover, we must have removed $e$ in some iteration
  between $t_0$ and the end, and hence $e \in \del([t_0 \ldots \infty))
  \sse \del(\infty)$.

\item More than than $3N_t/8$ edges in $L_0$ are not heads of any set in
  $\E(\infty)$. Take one such edge $e_0$ --- the sets in $\E(t_0)$ are
  singleton sets and hence $e_0$ is the head of the set $D_{t_0}(e_0)$.
  Let $t$ be the first (relevant) iteration such that $e_0$ is not the
  head of the set containing it in $\E(t)$, and suppose $e_0 =
  \head(D_{t-1}(e'))$ for some set $D_{t-1}(e') \in \E(t-1)$. In forming
  $\fcl{t}$, we must have short-cut $e'$ and some other edge $e''$ to
  form an edge $e \in \fcl{t}$.  Observe that $\length(\head(D(e''))
  \geq \length(e_0)$, else $e_0$ would continue to be the head of
  $D(e)$. Moreover,
  \[ \min\big(\pl(e'), \pl(e'')\big) \geq \min\big(\length(\head(e')),
  \length(\head(e''))\big) \geq \Delta_{t_0}/6.
  \]
  By the discussion in Claim~\ref{cl:forest-prop}, one of $e'$ and $e''$
  must lie on the cycle formed when we merged two supernodes in
  $\fcl{t-1}$, as in Step~4 of \updateforest. Further, if $e_t$ was the
  edge removed from this cycle, by the rule in Step~5 we get that the
  potential $\pl(e_t)$ is the maximum potential of any edge on this
  cycle, and hence $\pl(e_t) \geq \min(\pl(e'), \pl(e'')) \geq
  \Delta_{t_0}/6$. Hence we want to ``charge'' this edge $e_0 \in L_0$
  to $e_t \in \del(\infty)$ (which has  potential at least $\Delta_{t_0}/6$). However, up to three
  edges from $L_0$ may charge to $e_t$: this is because there can be at
  most three short-cut operations in any iteration (one from Step~3 and
  two from Step~6).
\end{itemize}
In both cases, we've shown the presence of at least $N_{t_0}/8$ edges in
$\del(\infty)$ of potential $\Delta_{t_0}/6$, which completes the proof
of Theorem~\ref{thm:charge}. \qedsymb


\subsection{An Extension of the Analysis in Section~\ref{sec:analysis}}
\label{sec:extension}

Let us now abstract out some properties used in the above analysis, so
that we can generalize the analysis to a broader class of algorithms for
\stf. This abstraction is used to show that variants of the above
algorithm, which are presented in Section~\ref{sec:timed-version} and in
Appendix~\ref{sec:projected}, are also $O(1)$-approximations.

Consider an algorithm $\A$ which maintains a set of supernodes, where a
supernode corresponds to a set of terminals, and two different
supernodes correspond to disjoint terminals.  Initially, we have one
supernode for each terminal. Further, a supernode could be active or
inactive. Once a supernode becomes inactive, it stays inactive. Now, at
each iteration, the algorithm picks two active supernodes, and replaces
them by a new supernode which is the union of the terminals in these two
supernodes (the new supernode could be active or inactive).  Note that
the iteration when a supernode becomes inactive is arbitrary (depending
on the algorithm $\A$).

As in the case of gluttonous algorithm, let $\C^f$ be the final
clustering produced by the algorithm $\A$, and $\Ts$ be a {\em tree}
solution to a \stf instance $(\D, \M)$.  Let $\clus{t}$ be the set of
supernodes at the beginning of iteration $t$ of $\A$. For an iteration
$t$, let $\delta_t$ be the minimum distance (in the metric
$\M/\clus{t}$) between any two active supernodes in
$\clus{t}$. Claim~\ref{clm:close-mono} gives the following fact.
\begin{fact}
  \label{fact:gen}
  The quantity $\delta_t$ forms an ascending sequence with respect to $t$. 
\end{fact}

Now Theorem~\ref{thm:one-tree-is-cheap} generalizes
to the following stronger result. 
\begin{corollary}
\label{cor:one-tree}
For any tree solution $\Ts$ to an instance $\I$, 
$\sum_{t} \delta_t \leq  48 \cdot \cost(\Ts).$
\end{corollary}

An important remark: this corollary is not making any claim about the
\emph{merging cost} of $\A$; at any iteration $\A$ could be connecting two
active supernodes which are much farther apart than $\delta_t$.


\section{A Timed Greedy Algorithm}
\label{sec:timed-version}
We now give a version of the gluttonous algorithm \timedglut where
supernodes are deemed active or inactive based on the current time and
not whether the terminals in the supernode have paired up with their
mates.\footnote{Timed versions of the primal-dual algorithm for Steiner
  forest had been considered previously in~\cite{GKPR, KLS05}; our
  version will be analogous to that of K\"onemann et al.~\cite{KLS05}
  which were used to get cross-monotonic cost-shares for Steiner
  forest.} This version will be useful in getting a strict cost-sharing
scheme.

The algorithm \timedglut is very similar to the gluttonous algorithm
except for what constitutes an active supernode.
\short{ Again, we shall maintain clustering of terminals into supernodes; and at each iteration $t$,
we shall merge the two closest active supernodes in the current metric. This will ensure that
the merging distances are monotone(~\autoref{clm:merge-mono}). Therefore, we can divide the
iterations of the algorithm into {\em stages}, where stage $i$ denotes the iterations where
the merging distance lies in the range $[\cee^i, \cee^{i+1})$.

 For a terminal $s$,
define $ \level(s)$ as $\lceil \log_\cee d_\M(s, {\bar s}) \rceil$ (recall that $\M$ denotes the
original metric). For a supernode $S$, define its {\em leader} as the
terminal in $S$ with highest level; in case of ties, choose the terminal with the
smallest index among these. A supernode $S$ will be active during stage $i$ of the algorithm
if the leader of $S$ has level at least $i$ (otherwise it is declared inactive).
Note that a supernode $S$  can remain active even if
for every terminal in it, its mate also lies in $S$. However, we can show that even this algorithm
has constant approximation ratio.
}

\full{
We  will again maintain a clustering of terminals (into supernodes) --
let $\clus{t}$ be the clustering at the beginning of iteration $t$.
Initially, at iteration $t=1$, $\clus{1}$ is the
 trivial clustering (consisting
of singleton sets of terminals). We maintain a set of edges $E'$ will be the
set of edges bought by the algorithm. Initially, $E' = \emptyset$.

We shall use $\Delta_t$ to denote the closest distance (in the metric $\M/\clus{t}$)
between two active supernodes in $\clus{t}$. Our algorithm will only merge active supernodes,
and an inactive supernode will not become active in future iterations. It follows that
$\Delta_t$ cannot decrease with $t$ (Fact~\ref{fct:dist-mono}). This allows us to divide the execution
of the algorithm into {\em stages}. Stage $i$ consists of those iterations $t$ for which $\Delta_t$ lies
in the range $[\cee^i, \cee^{i+1})$ (the initial stage belongs to stage 0, because we can assume w.l.o.g.\  that
the minimum distance between the terminals is 1).

For a terminal $s$,
define
\begin{gather}
  \level(s) := \lceil \log_\cee d_\M(s, {\bar s}) \rceil.
\end{gather}
Note that distances in this definition are measured in the original
metric $\M$. For a supernode $S$, define its {\em leader} as the
terminal in $S$ whose distance to its mate is the largest (and hence has
the largest level); in case of ties, choose the terminal with the
smallest index among these.

We shall use $\stagecl{i}$ to denote the clustering at the beginning of stage $i$ (note the
change in notation with respect to the clustering at the beginning of an iteration $t$, which will
be denoted by $\clus{t}$. So, if $t_i$ denotes the first iteration of stage $i$, then $\stagecl{i}$
is same as $\clus{t_i}$).
Now we specify when a supernode becomes inactive.
A terminal $s$ is {\em active} at the beginning of stage $i$ if $\level(s) \geq i$.
 A supernode $S$ will be  \emph{active} at the beginning of  a stage $i$ if $\level(\leader(S)) \geq
i$. Observe that supernodes do not become inactive {\em during} a stage -- if a terminal is active
at the beginning of a stage, it remains active during each of the iterations in this stage.

By the definition of a stage, the algorithm will satisfy the
invariant that the distance between any two active supernodes in $\stagecl{i}$
(in the metric $\M/\stagecl{i}$) is at least $\cee^i$. During stage $i$, the algorithm
repeatedly performs the following steps in each iteration $t$: pick any two
arbitrary pair of active supernodes $S', S''$ which are at most $\cee^{i+1}$ apart
(in the metric $\M/\clus{t}$). Further, we take any such $S'$-$S''$ path of length at
 most $\cee^{i+1}$ (in the graph induced by the metric $\M/\clus{t}$ on the vertex set
 $\clus{t}$) and add the edges (which go between supernodes) to $E'$.

Stage~$i$ ends when the merging distance between all remaining active
supernodes is at least $\cee^{i+1}$. Observe that when the algorithm
stops, we have a feasible solution---indeed, each terminal $s$ will
merge with its mate ${\bar s}$ by the end of stage $\level(s)$. At the
end, output a maximal acyclic subgraph of $E'$.

\full{
The analysis of \timedglut goes along the same lines as that of the
gluttonous algorithm. The analog of Theorem~\ref{thm:one-tree} is as
follows:
\begin{theorem}
  \label{thm:one-tree-timed}
  Let $\Fs = \{T_1^\star, T_2^\star, \ldots, T_p^\star\}$ be an optimal
  solution to the \stf instance $\I = (\M, \D)$.  Let clustering $\C^f$
  be produced by some run of the \timedglut algorithm. There exists
  another solution $\Fss$ for instance $\I$ such that
  \begin{OneLiners}
  \item[(a)] $\cost(\Fss) \leq \cost(\Fs) + \ceesq \cdot \width(\Fs) \leq
    \ceesqplusone \cdot \cost(\Fs)$, and
  \item[(b)] $\Fss$ is faithful to the clustering $\C^f$.
  \end{OneLiners}
\end{theorem}

\begin{proof}
  The proof is very similar to that of Theorem~\ref{thm:one-tree}, where
  we look over the run of \timedglut again to alter $\Fs$ into
  $\Fss$. Since \timedglut makes some arbitrary choices, we make the
  same choices consistently in this proof. We
  ensure very similar invariants:
  \begin{itemize}
  \item[(A)] The cost of edges in $\Fss \setminus \F^\star$ is at most $
    \ceesq (\width(\F^\star) - \width(\Fss))$, and
  \item[(B)] at any point in time, the forest $\Fss$ is faithful to the
    current clustering $\C$.
  \end{itemize}
  Observe the extra factor of $\ceesq$ in invariant~(A). Again, let two active
  supernodes $S'$ and $S''$ be merged in some stage $i$, and let $u$ and
  $v$ be the leaders of these supernodes respectively. The argument in
  Case~I remains unchanged. In Case~II, let $T_1^{\star\star}$ and
  $T_2^{\star\star} $ be the trees containing $u$ and $v$
  respectively. Being in stage $i$, we know that $\level(u), \level(v)
  \geq i$, since they are both still active, and that the distance
  between $S'$ and $S''$ in the current metric is at most $\cee^{i+1}$, since
  all merging costs in stage~$i$ lie between $\cee^i$ and $\cee^{i+1}$. So the
  cost of connecting $T_1^{\star\star}$ and $T_2^{\star\star}$ is at
  most
  \begin{align*}
    \cee^{i+1} \leq \min(\cee^{\level(u)+1}, \cee^{\level(v)+1}) &\leq
    \ceesq \cdot \min(d_{T_1^{\star\star}}(u, \bar{u}),
    d_{T_2^{\star\star}}(v, \bar{v})) \\
&\leq \ceesq \cdot
    \min(\width(T_1^{\star\star}), \width(T_2^{\star\star})).
  \end{align*}
  The rest of the argument remains unchanged.
\end{proof}
}
}
\begin{theorem}
  \label{thm:TG-approx}
  The \timedglut algorithm is a $\gammatg$-approximation algorithm for
  Steiner forest, where \full{$\gammatg=96 \times 5 = 480$.} \short{$\gammatg=480$.}
\end{theorem}

\full{
\begin{proof}
   Consider a solution $\Fss$ which is faithful with respect to the final clustering
   produced by the \timedglut algorithm. Suppose there are $m_i$ iterations during
   stage $i$. Then the total merging cost of the algorithm is at most $\sum_i \cee^{i+1} \cdot m_i$.

   We would like to use Corollary~\ref{cor:one-tree}. Let $\Fss$ consist of the trees
   $\Tss_1, \ldots, \Tss_k$. For a tree $\Tss_r$, and a stage $i$, let $I_{i,r}$ denote the
   iterations when we merge two supernodes with terminals belonging to the tree $V(\Tss_r)$ (note that
   the faithfulness property implies that there will be such a tree for each iteration of the
   algorithm). Let $m_{i,r}$ denote the cardinality of $I_{i,r}$. Clearly, $\sum_r m_{i,r} = m_r$.
   For an iteration $t$, and index $r$, let $\clusr{t}{r}$ denote the supernodes in $\clus{t}$ with
   terminals belonging to $V(\Ts_r)$. Define $\delta_{t,r}$ as the closest distance (in the metric
   $\M/\clus{t}$) between any two active supernodes with terminals belonging to $V(\Ts_r)$. If the iteration
   belongs to stage $i$, then $\delta_{t,r} \geq \cee^i$. Using Corollary~\ref{cor:one-tree}, we get
   $$ \sum_i \cee^{i+1} m_i \leq \cee \cdot \sum_r \sum_i \cee^i \cdot m_{i,r} \leq 98 \cdot \sum_r \cost(\Ts_r).$$
   The result now follows from Theorem~\ref{thm:one-tree-timed}.
\end{proof}

\subsection{An Equivalent Description of \timedglut}
\label{sec:equiv-timed}

An essentially equivalent way to state the \timedglut algorithm is as
follows. For a stage $i$, let $\M_i$ denote the metric $\M/\stagecl{i}$
corresponding to the clustering at the beginning of stage $i$. Construct
an auxiliary graph $\stageg{i}$ with vertex set being
the set of supernodes in $\stagecl{i}$, and edges between two vertices if the
two corresponding supernodes are active and the
distance between them is at most $\cee^{i+1}$ in
the metric $\M_i$. Pick a maximal acyclic set of edges $\pairs{i}$ in
this auxiliary graph $\stageg{i}$.
\begin{OneLiners}
\item For each edge $(S_1, S_2) \in \pairs{i}$, merge the supernodes
  $S_1, S_2$. Hence the clustering $\stagecl{i+1}$ at the end of stage $i$ is
  obtained by merging together all the supernodes that fall within a
  connected component of the subgraph $(\stageg{i}, \pairs{i})$.
\item For each edge $(S_1, S_2) \in \pairs{i}$, add edges corresponding
  to a path of length $< \cee^{i+1}$ in $\M_i$ to a set of edges $E_i$.
\end{OneLiners}
Finally, output a maximal sub-forest of the edges $\cup_i E_i$ added during this
process.

 One can now check that this algorithm is equivalent to the
\timedglut algorithm as described above; the key observation is that
because of the definition of the timed algorithm, an active terminal in
stage~$i$ stays active throughout the stage, and does not become
inactive partway through it. More formally, we have the following observation.
\begin{fact}
\label{fact:equiv}
Consider an execution of the \timedglut algorithm on an input $\I$. Then one can define  graphs $\stageg{i}$
and $\pairs{i}$ for each stage $i$ such that the set of supernodes at the beginning of stage $i$ in the
above algorithm is same as that of the \timedglut algorithm. Further, the two algorithms pick
the same set of edges in each stage.
\end{fact}
}

\section{Cost Shares for Steiner Forest}
\label{sec:cost-shares}

A \emph{cost-sharing} method is a function $\chi$ mapping triples of the
form $(\M,\D, (s, {\bar s}))$ to the non-negative reals, where $(\M,\D)$
is an instance of the Steiner forest problem, and $(s,{\bar s}) \in
\D$. We require the cost-sharing method to be \emph{budget-balanced}:
if $\Fs$ is an optimal solution to the instance $(\M,\D)$ then
\begin{align}
  \label{eq:cost-share}
  \sum_{ (s, {\bar s}) \in \D} \chi(\M,\D, (s, {\bar s})) \leq \cost(\Fs).
\end{align}

\full{
We will consider \emph{strict} cost-shares; these are useful for several
problems in network design (see details in the introduction).  There are
 two versions of
strictness: uni-strictness, and strictness. Uni-strict cost-shares for
Steiner forest were given by~\cite{GKPR, FKLS10}, whereas \gs
cost shares for Steiner forest have remained an open problem. We show
how to get both using the \timedglut algorithm.

\subsection{Uni-strict Cost Shares for Steiner Forest}
\label{sec:strict-cs}

\begin{definition}
\label{def:strict}
  Given an $\alpha$-approximation algorithm $\A$ for the Steiner forest problem, a cost sharing
  $\chi$ is called \emph{$\beta$-\us} with respect to $\A$ if for all
  demand pair $(s, {\bar s})$, the cost share $\chi(\M,\D,(s, {\bar s}))$ is at
  least $1/\beta$ times the distance between $s$ and $\bar s$ in the
  graph $G/F$, where $F$ is the forest returned by algorithm $\A$ on the
  input $(\M,\D-\{s, {\bar s}\})$.
\end{definition}
Our objective is to find an algorithm $\A$ and the associated cost share $\chi$ such that the parameters
$\alpha$ and $\beta$ are both constants.
\subsubsection{Defining $\chi$ and $\A$}
\label{sec:csot-share-def}

Let the constant $\gammatg$ denote the approximation ratio of the
algorithm \timedglut. The cost-sharing method is simple: for a terminal
$s$, let $\ell_s$ be the largest value such that $s$ is a leader in
stage $\ell_s$ \emph{and} its supernode is merged with some other
supernode during this stage (note that a supernode can go from being active
in the beginning of a stage to becoming inactive in the next stage without merging
with any supernode; this can happen because all terminals in it become inactive
in the next stage). Then
\begin{gather}
  \chi(\M,\D, (s, {\bar s})) :=
  \frac{\cee^{\ell_s} + \cee^{\ell_{\bar s}}}{2\gammatg}. \label{eq:chi}
\end{gather}

The algorithm $\A$ is a slight variant on the \timedglut
algorithm. Given an instance $(\M,\D)$, run the algorithm \timedglut on
this instance to get forest $F$. Now merge some of the trees in $F$ as
follows. Recall that the width of each tree $T$ in $F$ is defined to be
$\width(T) := \max_{(s, {\bar s}) \in T} d_T(s, {\bar s})$, where
$d_T(s, {\bar s})$ denotes the distance between $s$ and $\bar s$ in the
tree $T$.
While there are trees $T_1, T_2 \in F$ such that $d_{\M/F}(T_1, T_2)
\leq 5 \min(\width(T_1), \width(T_2))$, connect $T_1, T_2$ by a path
of length $d_{\M/F}(T_1, T_2)$ to get a tree $T$, and update $F \gets (F
\setminus \{T_1, T_2\}) \cup \{T\}$. Here, $d_{\M/F}(T_1, T_2)$ denotes the
minimum over all pairs $u \in T_1, v \in T_2,$ of $d_{\M/F}(u,v)$.

\subsubsection{Analysis}

We now prove that the cost sharing method $\chi$ is $\beta$-\us
with respect to $\A$, where $\beta$ is a constant. Recall that $F$ denotes
the forest returned by the algorithm $\A$.

To begin, observe
that the algorithm $\A$ is also a constant-factor approximation.

\begin{lemma}
\label{lem:A-is-good}
The algorithm $\A$ is an $6\gammatg$-approximation for Steiner forest.
\end{lemma}

\begin{proof}
  Let $F'$ be the forest returned by the \timedglut algorithm (called by $\A$).
  Consider the potential $\sum_{T \in F'} ( c(T) + 5\, \width(T))$. Since
  the width of each tree is at most the cost of its edges, and since
  \timedglut was a $\gammatg$-approximation, this potential is at most
  $6\gammatg$ times the optimal cost. Now, observe that whenever $\A$ merges two trees
  of this forest, the potential of the new forest does not increase. Therefore, the
  potential of the forest $F$ is also at most $6 \gammatg$ times the optimal cost.
\end{proof}

\begin{lemma}
  \label{lem:cost}
  The function $\chi$ is a budget-balanced cost sharing method.
\end{lemma}

\begin{proof}
  We need to prove the inequality~(\ref{eq:cost-share}). To do this, let
  us run \timedglut and ``charge'' the cost of merging two active
  supernodes to the leaders of the respective clusters---charge half of
  the distance between these two supernodes to the leaders of each of
  these supernodes.  Clearly, the total charge assigned to the terminals
  is equal to the total cost paid by the algorithm \timedglut, which at most
  $\gammatg\, \cost(\Fs)$. Finally, we make the observation that each
  terminal $s$ is charged at least $\gammatg$ times the cost share
  (since it is charged at least $\cee^{\ell_s}/2$ in  stage $\ell_s$)
  to complete the proof.
\end{proof}

To prove the uni-strictness property, fix a terminal pair $(s, {\bar s})$,
and consider two instances: $\I=(\M,\D)$ and $\I'=(\M,\D-\{(s,{\bar
  s})\})$. For instance $\I$, let $\stagecl{i}$ denote the set of
supernodes at the beginning of stage $i$; let $\stagecl{'i}$ be the
corresponding set for $\I'$. Let $\M_i$ and $\M_i'$ denote the metrics
$\M/\stagecl{i}$ and $\M/\stagecl{'i}$ respectively. Recall that
$\level(s) = \lceil \log_\cee d_\M(s, {\bar s}) \rceil$.

The following claim will be convenient to understand the behavior of
\timedglut.
\begin{lemma}
  \label{lem:timed-supernode}
  Consider stage $i$ in the execution \timedglut on the instance
  $\I$. Define a graph $\G_i$ on the vertex set $\stagecl{i}$, with an edge
  between two active supernodes $C_1, C_2 \in \stagecl{i}$ if there is a
  path of length at most $\cee^{i+1}$ between them in $\M_i$ that does not
  contain any other active supernode as an internal node. If the
  connected components of $\G_i$ are $H_1, \ldots, H_q$, then
  $\stagecl{i+1}$ has $q$ supernodes---one supernode for each $H_j$ (formed
  by merging the supernodes in $H_j$).
\end{lemma}
\begin{proof}
The statement is essentially the same as Fact~\ref{fact:equiv} except that in the
graph $\stageg{i}$ (defined in Section~\ref{sec:equiv-timed}), we join two active supernodes
$C_1, C_2 \in \stagecl{i}$ by an edge if the distance between them in the metric
$\M_i$ is at most $\cee^{i+1}$, whereas here in the graph $\G_i$, we wish
to have a path of length at most $\cee^{i+1}$ with no internal vertex being an active supernode.
We claim that the connected components in the two graphs are the same, and hence, the statement
in the lemma follows.

Clearly, an edge $e \in \G_i$ is present in $\stageg{i}$ as well. Now, consider an edge $(C_1, C_2)$
in $\stageg{i}$. Let $P$ be the shortest path of
length at most $\cee^{i+1}$ between $S_1$ and $S_2$ in the metric $\M_i$. Let the active supernodes
on this path be $C_1=C_{i_1}, C_{i_2}, \ldots, C_{i_p}=C_2$ (in this order). Then $\G_i$ has edges $(C_{i_r}, C_{i_{r+1}})$
for $r=1, \ldots, p-1$. Therefore $C_1$ and $C_2$ are in the same connected component of $\G_i$. This proves
the desired claim.
\end{proof}

\begin{theorem}[Nesting]
  \label{thm:nesting}
  For $i \leq \level(s)$, let $C_s$ and $C_{\bar s}$ be the supernodes
  in $\stagecl{i}$ containing $s$ and $\bar s$ respectively. The following
  hold:
  \begin{enumerate}
  \item[(a)] If $C_s \neq C_{\bar s}$, we can arrange the supernodes in
    $\stagecl{'i}$ as $C_1', \ldots, C_p'$ such that $C_s = \{s\} \cup
    C_1' \cup \ldots \cup C_a',$ $C_{\bar s}= \{ \bar s\} \cup C_{a+1}'
    \cup \ldots \cup C_b'$ for some $0 \leq a \leq b \leq p$. Moreover,
    $\stagecl{i}-\{C_s,C_{\bar s}\} = \{C_{b+1}', \ldots, C_p'\}$.  If
    $C_s=C_{\bar s}$, we can arrange the supernodes in $\stagecl{'i}$
    as $C_1', \ldots, C_p'$ such that $C_s = \{s, {\bar s}\} \cup C_1' \cup \ldots
    \cup C_b',$ for some $0 \leq b \leq p$. Also, $\stagecl{i}-\{C_s\} =
    \{C_{b+1}', \ldots, C_p'\}$.

  \item[(b)] Suppose $C_s$ and $C_{\bar s}$ are distinct supernodes. Then
  for any terminal $v \in C_s$,  $d_{\M'_i}(s, v) \leq
    2 \cdot \cee^{i}$.  Similarly, for any $v \in C_{\bar s}$,
    $d_{\M'_i}(\bar s, v) \leq 2 \cdot \cee^{i}$.
  \item[(c)] Suppose $C_s = C_{\bar s}$. Then, $d_{\M'_i}(s, {\bar s})
  \leq 4 \cdot \cee^{i}$.

  \end{enumerate}
\end{theorem}

\begin{proof}
  We induct on $i$. At the beginning, all clusters are singletons, so
  the base case is easy. For the inductive step, suppose the statement
  of the theorem is true for some $i < \level(s)$. Assume that $C_s \neq
  C_{\bar s}$, the other case is similar. Apply
  Lemma~\ref{lem:timed-supernode} to stage $i$ in both $\I$ and $\I'$,
  and let $\G_i$ and $\G_{i}'$ be the corresponding graphs on the vertex
  sets $\stagecl{i}$ and $\stagecl{'i}$ (as defined in Lemma~\ref{lem:timed-supernode}).
  We know that
  the supernodes of $\stagecl{i+1}$ and $\stagecl{'i+1}$ correspond to the
  connected components of these graphs; we now use this information to
  prove the induction step.

  By the induction hypothesis, the supernodes in $\stagecl{i}$ can be
  labeled $C_{s}, C_{\bar s}, C_{b+1}', \ldots, C_p'$; moreover, we can
  define a map $\phi: V(\G_i') \to V(\G_i)$ as follows:
  \begin{gather}
    \phi(C_j') :=
    \begin{cases}
      C_s & 0 \leq j \leq a \\
      C_{\bar s} & a+1 \leq j \leq b \\
      C_j' \quad &  b+1 \leq j \leq p
    \end{cases}
  \end{gather}
  Suppose there is an edge between $C_j'$ and $C_k'$ in $\G_i'$. By the
  definition of $\G_i$, both are active supernodes, and the length of
  the shortest path between them in the metric $\M_i'$ is at most
  $\cee^{i+1}$. This path has no greater length in the metric $\M_i$, since
  the supernodes in $\stagecl{i}$ are unions of supernodes in
  $\stagecl{'i}$. This means there is a path between $\phi(C_j')$ and
  $\phi(C_k')$ in $\G_i$, i.e., the clustering $\stagecl{'i+1}$ is
  a refinement of $\stagecl{i+1}$ (because these clusterings are determined by
   the connected components of the corresponding graphs).

  Now consider an edge $e$ in $\G_i$. For the first part of the theorem,
  suppose $e=(C_j', C_k')$ where both $j,k \geq b+1$. If $P$ is the
  corresponding path in $\M_i$ between these two supernodes, then $P$ cannot
  contain $C_s$ or $C_{\bar s}$ as an internal node (because it does not
  contain any active supernodes as internal nodes, and both $C_s,
  C_{\bar s}$ are active). But then the length of $P$ remains unchanged
  in $\M_i'$, and we have the corresponding edge $(C_j', C_k')$ in
  $\G_i'$ as well. This means that all the connected components of
  $\G_i$ not containing $C_s$ or $C_{\bar s}$ also form connected
  components in $\G_i'$. Combined with the fact that $\stagecl{'i+1}$
  is a refinement of $\stagecl{i+1}$, this proves the part~(a) of the
  theorem for the case $C_s \neq C_{\bar s}$. (The proof for the other
  case is similar.)

  For part~(b),  let $H_1$ be the connected component of $\G_i$ which
  contains the supernode $C_s$ (as a vertex). So all the supernodes
  in $H_1$ will merge to form a single supernode of $\stagecl{i+1}$.
  As argued in the paragraph above, any
  edge in $H_1$ which is not incident with $C_s$ is also
  present in $\G_i'$(recall
  that we are assuming $C_{\bar s}$ is not one of the vertices in $H_1$).
  Let $v$ be a terminal in a supernode $B$ in $H_1$. Let the path from
  $C_s$ to $B$ in $H_1$ be $C_s= A_0, A_1, \ldots, A_r = B$. Since the edges
  $(A_1, A_2), \ldots, (A_{r-1}, A_r)$ belong to $\G_i'$ as well,
  $A_1, \ldots, A_r$ will lie in the same supernode in $\stagecl{'i+1}$. Therefore,
  $d_{\M_{i+1}'}(s, v) \leq
  \cee^{i+1} + 2 \cdot \cee^i = 2 \cdot \cee^{i+1}$, where the term
  $2 \cdot \cee^i$ is present to account for the distance between $s$ and the terminal
  in $C_s$ which is closest to $A_1$   -- this distance can be at most $2 \cdot \cee^i$ by the
  induction hypothesis.

  For part~(c), consider the last stage $i$ such that the supernodes $C_s$ and
  $C_{\bar s}$ are distinct (so the result in part~(b) applies to this stage).
  The same argument as above applies except that when we consider the path from
  $C_s$ to $C_{\bar s}$ in the component of $\G_i$ containing them, we will
  have to account for the first and the last edges in this path.
\end{proof}

We are now ready to prove the uni-strictness of the cost-shares.  We
run \timedglut on the instance $\I$ to get the cost shares, and let the
cost share for $(s, \bar s)$ be as in~(\ref{eq:chi}).
Now let $F'$ be the forest returned by the algorithm $\A$ on the
instance $\I'$. Recall that  $\M/F'$ denotes the metric
$\M$ with the connected components in $F'$ contracted to single points.
\begin{lemma}
  \label{lem:costshare}
  The distance between $s$ and $\bar s$ in $\M/F'$ is at most $4  \cdot
  (\cee^{\ell_s+1} + \cee^{\ell_{\bar s}+1})$.
\end{lemma}

\begin{proof}
  Let $\ti := \max(\ell_s, \ell_{\bar s}) + 1$.  Suppose $\ti \geq
  \level(s)$, then the claim is trivial because $d_\M(s,\bar s) \leq
  \cee^{\level(s)} \leq \cee^{\ti} \leq \cee \cdot (\cee^{\ell_s}
  +\cee^{\ell_{\bar s}})$; hence consider the case where $\ti \leq
  \level(s) - 1$.

  Let $C_s$ and $C_{\bar s}$ denote the supernodes containing $s$ and $\bar s$
  in clustering $\stagecl{j}$ respectively.
  There are two cases. The first case is when $C_s$ is same as $C_{\bar s}$.
  In this case, part~(c) of Theorem~\ref{thm:nesting} implies that
  $\smash{d_{\M'_{\ti}}(s,\bar s)} \leq 4 \cdot \cee^{\ti} \leq 4 \cdot
   (\cee^{\ell_s+1} +\cee^{\ell_{\bar s}+1})$. The distance in metric
  $d_{\M/F'}$ can only be smaller.

  The other case is when $C_s$ and $C_{\bar s}$ are different. Note that $C_s$
  will be merged with another supernode in some stage during or after stage $j$
  (eventually the $s$ and ${\bar s}$ will end up in the same supernode). Since
  $j > \ell_s$, it follows from the definition of $\ell_s$ that $s$ is not the
  leader of $C_s$. Similarly, ${\bar s}$ is not the leader of $C_{\bar s}$.
  Let the leaders of $C_s$ and $C_{\bar s}$ be $v_1$ and $v_2$ respectively.
   By Theorem~\ref{thm:nesting}(b), we know that
  $\smash{d_{\M'_{\ti}}(s,v_1)} \leq 2  \cdot \cee^{\ti}$ and
   $d_{\M'_{\ti}}({\bar s}, v_2 ) \leq 2 \cdot \cee^{\ti}$. Consequently,
  \begin{gather*}
    d_{\M'_{\ti}}(v_1, v_2) \leq d_{\M'_{\ti}}(v_1, s) +
    d_{\M'_{\ti}}(s, \bar s) + d_{\M'_{\ti}}({\bar s}, v_2) \leq
    4\cdot\cee^{\ti} + d_\M(s, \bar s) \leq 5\, d_\M(s, \bar s),
  \end{gather*}
   where the last inequality follows because $d_{\M}(s, \bar s)
   \geq d_{\M_\ti}(C_s, C_{\bar s}) \geq \cee^j$.

  Let $F''$ be the final forest produced by $\timedglut$ on the instance
  $\I'$; recall that $F'$ is obtained from $F''$ by merging together
  some of these trees. Let $T_1$ and $T_2$ be the trees in $F''$ which
  contain $v_1$ and $v_2$ respectively. Since the distance between $v_1$
  and $v_2$ is already at most $5\, d_\M(s, \bar s)$ at the
  beginning of stage $\ti$, we know that $d_{\M/F'}(T_1, T_2) \leq 5\, d_\M(s, \bar s)$, where $\M/F'$ denotes the metric $M$ with
  the trees in $F'$ contracted.

  Since $s$ lost its leadership to $v_1$, it must be the case that $d(s,
  {\bar s}) \leq d(v_1, {\bar {v_1}})$; thus $\width(T_1) \geq d(s, \bar
  s)$; a similar argument shows $\width(T_2) \geq d(s, \bar s)$.
  Since $d_{\M/F'}(T_1, T_2) \leq 5\, \min(\width(T_1), \width(T_2)) $,
  the algorithm $\A$ would have merged $T_1$ and $T_2$ into one
  tree. This makes the distance $d_{\M/F'}(v_1, v_2) = 0$ and hence
  $$d_{\M/F'}(s, \bar s) \leq d_{\M/F'}(s, v_1) + d_{\M/F'}(v_2, {\bar s})
  \leq d_{\M/\stagecl{\ti}}(s, v_1)+ d_{\M/\stagecl{\ti}}({\bar s}, v_2)
  \leq
  4 \cdot \cee^{\ti},$$ proving the claim.
\end{proof}

This shows that the cost of connecting $(s, \bar s)$ in $\M/F'$ is at
most $\beta := 16\gammatg$ times the cost share of $(s, \bar s)$,
which proves the uni-strictness property.

\subsection{Strict Cost Shares}
\label{sec:group-strict}

We now extend the previous cost sharing scheme to the more general
\emph{\gs} cost sharing scheme. Let $\chi$ be a
budget-balanced cost sharing function for the Steiner forest problem. As
before, let $\A$ be an $\alpha$-approximation algorithm for the Steiner forest problem.
}
\short{We now consider the notion of \gs cost shares (the details for the
construction of the simpler uni-strict cost shares are given in the full version).
}
\begin{definition}
  \label{def:goup-strict}
  A cost-sharing function $\chi$ is {\em $\beta$-\gs} with
  respect to an algorithm $\A$ if for all pairs of disjoint terminal
  sets $\D_1, \D_2$ lying in a metric $\M$, the following condition
  holds: if $\D$ denotes $\D_1 \cup \D_2$, then $\sum_{(s, {\bar s})
    \in \D_2} \chi(\M, \D, (s, {\bar s}))$ is at least $1/\beta$ times
  the the cost of the optimal Steiner forest on $\D_2$ in the metric
  $\M/F$, where $F$ is the forest returned by $\A$ on the input $(\M,
  \D_1)$.
\end{definition}

In addition to the \timedglut algorithm, we will also need a timed
primal-dual algorithm for Steiner forest, denoted by \TPD. The input for the
\TPD algorithm
is a set of terminals, each terminal $s$ being assigned an
\emph{activity time} $\time(s)$ such that the terminal is \emph{active}
for all times $t \leq \time(s)$. The primal-dual algorithm grows moats
around terminals as long as they are active and buys edges that ensure
that if two moats meet at some time $t$, all the terminals in these
moats that are active at time $t$ lie in the same tree. One can do this
in different ways (see, e.g.,~\cite{GKPR,Pal,KLS05}); for concreteness
we refer to the \texttt{KLS} algorithm of K\"onemann et al.~\cite{KLS05}
which gives the following guarantee:
\begin{theorem}
  \label{thm:KLS}
  If $\time(s) = \frac12 d_\M(s, \bar s)$ for all terminals $s$, then
  the total cost of edges bought by the timed primal-dual algorithm
  \texttt{KLS} is at most $2 \cdot \OPT(\I)$.
\end{theorem}
\full{The following property can be shown for the \texttt{KLS} algorithm:
\begin{lemma}
  \label{lem:kls-scaling}
  Multiplying the activity times by a factor of $K \geq 1$ to
  $\frac{K}{2} \cdot d_\M(s, \bar s)$ causes the
  \texttt{KLS} algorithm to output another feasible solution of total
  cost at most $2K \cdot \OPT(\I)$.
\end{lemma}

\subsubsection{Defining $\chi$ and $\A$}
\label{sec:gs-cs-defs}
}

\paragraph{Defining $\chi$ and $\A$}
To define the cost-shares for the instance $\I = (\M, \D)$, run the
algorithm \timedglut on $\I$.
\short{ During stage $i$, if we merge supernodes $S$ and $S'$ with leaders $s$ and $s'$
respectively, then we increment the cost-share of each of $(s, \bar s)$ and $(s', \bar s')$ by
$\frac{\cee^{i+1}}{2\gammatg}$. The budget balance property follows directly from the
analysis of the \timedglut algorithm.

}
\full{Recall the description of the algorithm as
given in \autoref{sec:equiv-timed}: in each stage $i$, we choose a
collection $\pairs{i}$ of pairs of supernodes whose mutual
distance (in the metric $\M/\stagecl{i}$) lies in the range $[\cee^i, \cee^{i+1})$, merge each such
pair of supernodes to get the new clustering (and add edges in the
underlying graph of at most as much length). For pair $(S, S') \in
\pairs{i}$, if $s, s'$ are the leaders of $S, S'$ respectively,
increment the cost-share of each of $(s, \bar s)$ and $(s', \bar s')$ by
$\frac{\cee^{i+1}}{2\gammatg}$. Since the analysis of the \timedglut
algorithm proceeds by showing that the quantity $\sum_i | \pairs{i} |
\cdot \cee^{i+1} \leq \gammatg \, \OPT(\I)$, the budget-balance property
follows.
}

The algorithm $\A$ on input $\I = (\M, \D)$ is simple: set the activity
time $\time(s) := 6 \cdot \cee^{\level(s)+1}$ for each terminal $s$, and
run the algorithm \TPD. \full{The following claim immediately follows from
Lemma~\ref{lem:kls-scaling} and Theorem~\ref{thm:KLS}, and the fact that
$6 \cdot \cee^{\level(s) + 1} \leq 12 \cdot \cee^2 \cdot \frac12 d_\M(s,\bar s)$.
}
\short{ The following claim follows from a simple extension of Theorem~\ref{thm:KLS}
and the fact that $6 \cdot \cee^{\level(s) + 1} \leq 48 \cdot \frac12 d_\M(s,\bar s)$.}
\begin{lemma}
  \label{lem:alg-timed}
  The algorithm $\A$ is a $96$-approximation algorithm
  for Steiner forest.
\end{lemma}


\full{\subsubsection{Proving Strictness}}
\short{\paragraph{Proving Strictness}
We now give an outline of the proof of the strictness property, leaving details to
the full version.}
Given a set of demands $\D$ in a metric space $\M$, and a partition into
$\D_1 \cup \D_2$, we run the algorithm $\A$ on $\D_1$---let $F_1$ be the
forest returned by this algorithm, and let metric $\M_1$ be obtained by
contracting the edges of $F_1$. To prove the strictness property, we now
exhibit a ``candidate'' Steiner forest $F_2$ for $\D_2$ in the metric
$\M_1$ with cost at most a constant factor times $\sum_{(s, {\bar s})
  \in \D_2} \chi(\M, \D, (s, {\bar s})),$ the total cost-share assigned
to the terminals in $\D_2$.

\short{ To define the forest $F_2$ connecting $\D_2$, we now imagine running
\timedglut on the entire demand set $\D_1 \cup \D_2$ on the original
metric $\M$, look at paths added by that algorithm, and choose a
carefully chosen subset of these paths to add to $F_2$. This is the
natural thing to do, since such a run of \timedglut was used to define
the cost-shares $\chi$ in the first place. Ideally, whenever this algorithm
connects two supernodes $S$ and $S'$, we will add a path between them
(in the forest $F_2$) unless there are terminals $v$ and $v'$  in $S$ and $S'$
respectively such that both $v$ and $v'$ lie in the same tree of $F_1$ (note that
$v, v'$ belong to $\D_1$).
However, analyzing such a scheme directly is quite challenging. Instead, we
carefully match the runs of the following algorithms: \timedglut on $\D_1 \cup \D_2$,
denoted by $\R$,
and $\A$ on $\D_1$. For any parameter $t$, the run of $\A$, which is a moat
growing algorithm, at time $t$ corresponds to the run $\R$ when the merging distance is
about $t$. We now show that if $\R$ merges supernodes $S$ and $S'$, then we will
add a subset of the edges on the shortest path between $S$ and $S'$ to the forest $F_2$
such that these edges correspond to the ``dual'' raised by the corresponding run of $\A$ (and
so would have contributed towards the cost share of a terminal).
 Details of the analysis
can be found in the full version.}

\full{
Recall that the algorithm $\A$ on $\D_1$ is just the \TPD algorithm. We
divide this algorithm's run into stages, where the $i^{th}$ stage lasts
for the time interval $[6 \cdot \cee^i, 6 \cdot \cee^{i+1})$; the $0^{th}$ stage lasts
for $[0, 6 \cdot \cee)$. Let $\forestone{i}$ be the edges of the output forest
$F_1$ which become tight during stage $i$ of this run, and
$\moatsone{i}$ be the set of moats at the beginning of stage $i$. These
moats are defined in the original metric $\M$.

\paragraph{Defining a Candidate Forest $F_2$}

To define the forest $F_2$ connecting $\D_2$, we now imagine running
\timedglut on the entire demand set $\D_1 \cup \D_2$ on the original
metric $\M$, look at paths added by that algorithm, and choose a
carefully chosen subset of these paths to add to $F_2$. This is the
natural thing to do, since such a run of \timedglut was used to define
the cost-shares $\chi$ in the first place. Recall the description of
\timedglut from \autoref{sec:equiv-timed}, and let $\R$ denote this run
of \timedglut on $\I = (\M, \D_1 \cup \D_2)$.

We examine the run $\R$ stage by stage: at the beginning of stage $i$,
the run $\R$ took the current clustering $\stagecl{i}$, built an auxiliary
graph $\stageg{i}$ whose nodes were the supernodes in $\stagecl{i}$ and
edges were pairs of active supernodes that had mutual merging distance
at most $\cee^{i+1}$, picked some maximal forest $\pairs{i}$ in this
graph, merged these supernode pairs in $\pairs{i}$, and bought edges in
the underlying metric corresponding to paths connecting these supernode
pairs. We show how to choose some subset of these underlying edges to
add to our candidate forest---we denote these edges by $\foresttwo{i}$.

In the following, we will talk about edges $(S,S') \in \pairs{i}$ (which
are edges of the auxiliary graph $\stageg{i}$) and edges in the metric
$\M$. To avoid confusion, we refer to $(S,S')$ as \emph{pairs} and those
in the metric as edges.

What edges should we add to $\foresttwo{i}$? For that, look at the end
of stage~$i$ of the run of $\A$ on $\D_1$; the primal-dual algorithm has
formed a set of moats $\moatsone{i+1}$ at this point. We define an
equivalence relation on the supernodes in $\stagecl{i}$ as follows. If the
leaders of two supernodes $S, S' \in \stagecl{i}$ both lie in $\D_1$, and
also in the same moat of $\moatsone{i+1}$, we put $S$ and $S'$ in the
same equivalence class. Now if we collapse each equivalence class by
identifying all the supernodes in that class, the pairs in $\pairs{i}$
may no longer be acyclic in the collapsed version of $\stageg{i}$, they
may contain cycles and self-loops. Consider a maximal acyclic set of
pairs in $\pairs{i}$, and denote the dropped pairs by
$\pairsdrop{i}$. The set $\pairs{i}$ is now classified into three parts~(see
Figure~\ref{fig:cost} for an example):
\begin{itemize}
\item Let $\pairsgood{i}$ be pairs $(S,S') \in \pairs{i} \setminus
  \pairsdrop{i}$ for which at least one of $\leader(S),\leader(S')$
  belongs to $\D_2$.
\item Let $\pairsbad{i}$ be pairs $(S,S') \in \pairs{i} \setminus
  \pairsdrop{i}$ for which both of $\leader(S)$ and $\leader(S')$ belongs
  to $\D_1$.
\item Of course, $\pairsdrop{i}$ is the set of pairs $(S,S') \in
  \pairs{i}$ dropped to get an acyclic set.
\end{itemize}

Given this classification, the edges we add to $\foresttwo{i}$ are as
follows. Recall that in the run $\R$, for each pair $(S,S')
\in \pairs{i}$, we had added edges connecting those two supernodes of total
length at most $\cee^{i+1}$. For each pair in $\pairs{i} \setminus
\pairsdrop{i}$, we now add the same edges to $\foresttwo{i}$. We
further classify these edges based on their provenance: the edges added
due to a pair $(S,S') \in \pairsgood{i}$ we call \emph{good edges}, and
those added due to a pair in $\pairsbad{i}$ we call \emph{bad
  edges}. Observe that we add no edges for pairs in $\pairsdrop{i}$.

\begin{figure}[h]
\begin{center}
  \includegraphics[height=35mm]{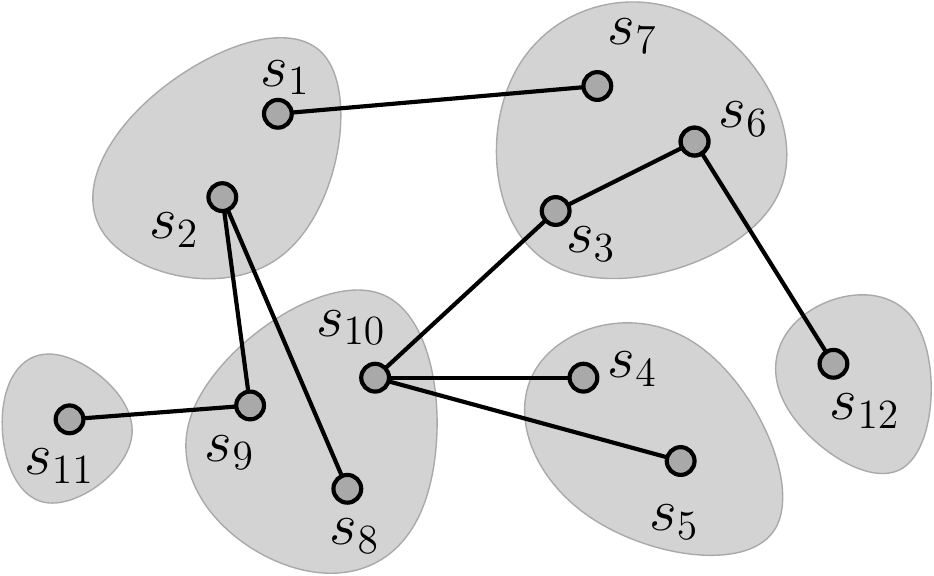}
\caption{\footnotesize The solid edges denote the pairs $\pairs{i}$, and the grey regions denote the equivalence classes.
Let $\pairsdrop{i}$ be the pairs $\{(S_2, S_9), (S_4, S_{10}), (S_3,S_{10}), (S_3, S_6) \}$. Assume that leaders of
$S_{11}$ and $S_{12}$ belong to $\D_2$ (the leaders of the rest of the supernodes must be in $\D_1$ because
the equivalence classes corresponding to these supernodes  have cardinality larger than 1). So,
$\pairsgood{i}$ consists of pairs $\{(S_6, S_{12}), (S_{11}, S_9)\}$ and $\pairsbad{i}$ consists of
 $\{(S_1, S_7), (S_2, S_8), (S_5, S_{10}) \}$. }
\label{fig:cost}
\end{center}
\end{figure}

This completes the construction of the set $\foresttwo{i}$.  The forest
$F_2$ is obtained by taking the union of $\cup_i \foresttwo{i}$. The
task now is to show (a)~feasibility, that the edges in $F_2$ form a
Steiner forest connecting up the demands of $\D_2$ in metric $\M_1$, or
equivalently that $F_1 \cup F_2$ is a Steiner forest on the set $\D_1
\cup \D_2$, and (b)~strictness, that the cost of $F_2$ is comparable to
the cost shares assigned to the demands in $\D_2$.

\paragraph{Feasibility}

First, we show feasibility, i.e., that $F_1 \cup F_2$ connects all pairs
in $\D_2$. Observe that we took the run $\R$, and added to $F_2$ some of
the edges added in $\R$. Had we added all the edges, we would trivially
get feasibility (but not the strictness), but we omitted edges
corresponding to pairs in $\pairsdrop{i}$. The idea of the proof
is that such supernodes will get connected due to the other connections, and
to the fact that we inflated the activity times in the \TPD
algorithm. Let's give the formal proof, which proceeds by induction over
time.

For integer $i$, define $\forestone{\leq i} := \cup_{j \leq i}
\forestone{j}$, and define $\foresttwo{\leq i}$ similarly.  The first
claim relates the stages in the run $\R$ of \timedglut$(\D_1 \cup \D_2)$
to the stages in the run of \TPD.
\begin{claim}
  \label{cl:corres}
  If terminal $s \in \D_1$ is active at the beginning of stage $i$ in
  the run $\R$, then the moat containing $s$ remains active during stage
  $i$ of the run of $\TPD$ on $\D_1$.
\end{claim}
\begin{proof}
  Since $s \in \D_1$ is active in stage $i$, $\level(s) \geq i$.  Hence
  its activity time $\time(s) \geq 4 \cdot \cee^{i+1}$. Since stage~$i$ for
  the timed primal-dual algorithm ends at time $4 \cdot \cee^{i+1}$, the moat
  containing $s$ in \TPD will be active at least until the end of stage $i$.
\end{proof}

\begin{lemma}
  \label{lem:moats}
  Let $\stagecl{i}$ be the clustering at the beginning of stage $i$ in the run
  $\R$. Then,
  \begin{OneLiners}
  \item[(a)] For any $S \in \stagecl{i}$, all terminals in $S$ lie in the
    same connected component of $\forestone{\leq i-1} \cup
    \foresttwo{\leq i-1}$.
  \item[(b)] For every $(S, S') \in \pairs{i} \setminus \pairsdrop{i}$,
    the terminals in $S \cup S'$ lie in the same connected component of
    $\forestone{\leq i-1} \cup \foresttwo{\leq i}$.
  \end{OneLiners}
\end{lemma}
\begin{proof}
  We first show that if the statement~(a) is true for some stage $i$,
  then the corresponding statement~(b) is also true (for this stage).
  Consider a pair $(S, S') \in \pairs{i}$. If $(S, S') \in \pairsgood{i}
  \cup \pairsbad{i}$, the
  edges we add to $\foresttwo{i}$ would connect the terminals in $S \cup S$, as long as
  all the supernodes in $\stagecl{i}$ formed connected components. But by
  the assumption, we know that edges in $\forestone{\leq i-1}
  \cup \foresttwo{\leq i-1}$ connect up each supernode in
  $\stagecl{i}$. Consequently, terminals in $S \cup S'$ lie in the same
  connected component of $\forestone{\leq i-1} \cup \foresttwo{\leq i}$.
  This proves statement~(b).

  We now prove statement~(a) by induction on $i$. At the beginning of stage $i = 0$,
  each supernode $S$ is a singleton and hence the statement is true.

  Now to prove the induction step for~(a). It suffices to show that if
  $(S, S') \in \pairsdrop{i}$ then $S \cup S'$ is contained in the same
  component in $\forestone{\leq i} \cup \foresttwo{\leq i}$. We
  distinguish two cases. The first case is when $S$ and $S'$ both lie in
  the same equivalence class that was used to construct
  $\pairsdrop{i}$. Then $s = \leader(S)$ and $s' = \leader(S')$ belong
  to $\D_1$ and also to the same moat in $\moatsone{i+1}$, at the end of
  stage $i$. Since $S, S'$ are active in stage $i$ of \timedglut($\D$),
  both $s, s'$ have level at least $i$. By Claim~\ref{cl:corres} they
  remain active throughout stage~$i$ of \TPD. Moreover, the end of that
  stage they share the same moat. Hence, $s, s'$ belong to the same moat
  when active---but recall that the \TPD algorithm ensures that whenever
  two active terminals belong to the same moat they lie in the same
  connected component. Hence $s,s'$ must lie in the same connected
  component of $\forestone{\leq i}$. By the induction hypothesis, the
  rest of $S, S'$ are connected to their leaders in $\forestone{\leq
    i-1} \cup \foresttwo{\leq i-1}$. Hence $S, S'$ are connected in
  $\forestone{\leq i} \cup \foresttwo{\leq i-1}$; indeed for each
  equivalence class, the supernodes that belong to it are connected
  using those edges.

  The second case is when for pair $(S,S') \in \pairsdrop{i}$, the
  supernodes $S$ and $S'$ do not fall in the same equivalence class, but
  the adding pair $(S,S')$ to $\pairs{i} \setminus \pairsdrop{i}$ would
  form a cycle when equivalence classes are collapsed. The argument here
  is similar: if the leaders are again $s, s'$, then the above arguments
  applied to each pair on the cycle, and to each equivalence class imply
  that $s$ and $s'$ must be connected in $\forestone{\leq i} \cup
  \foresttwo{\leq i-1}$---and therefore so must $S \cup S'$.
\end{proof}

Since each pair $\{s, \bar{s}\} \in \D_2$ is contained in some supernode
at the end of the run $\R$, Lemma~\ref{lem:moats} implies that they are
eventually connected in using $F_1 \cup F_2$ as well.  This completes
the proof that $F_1 \cup F_2$ is a feasible solution to the demands in
$\D_2$.

\paragraph{Bounding the Cost of Forest $F_2$}

Finally, we want to bound the cost of the edges in $F_2$ by a constant
times $\sum_{(s, \bar s) \in \D_2} \chi(\M, \D_1 \cup \D_2, (s, \bar
s))$. If $\countbad{i} := |\pairsbad{i}|$, then the total
cost of bad edges is at most
\begin{eqnarray}
  \label{eq:badcost}
  \sum_i \countbad{i} \cdot \cee^{i+1}
\end{eqnarray}
because the length of edges added for each connection in $\pairsbad{i}$
is at most $\cee^{i+1}$.

\begin{lemma}
  \label{lem:goodcost}
  The total cost of the edges in $\cup_i \foresttwo{i}$ is at least $3
  \sum_i \countbad{i} \cdot (\cee^{i+1}-\cee^i)$.
\end{lemma}
\begin{proof}
  For this proof, recall that we run the primal-dual process on the
  metric $\M$, and $\moatsone{i}$ are the dual moats at the beginning of
  stage $i$. Let $E_i$ denote the set of tight edges lying inside the
  moats in $\moatsone{i}$.  We prove the following statement by
  induction on $i$: the total cost of edges in $\foresttwo{\leq i} \cap
  E_i$ is at least $3 \sum_{j \leq i} \countbad{j} \cdot
  (\cee^{j+1}-\cee^{j})$.

  The base case for $i=0$ follows trivially because the $\foresttwo{0}$ is empty, and
 $\countbad{0}$ is also 0.

  Suppose the statement is true for some $i-1$. Now consider the pairs in
  $\pairsbad{i}$---these correspond to pairs of supernodes $(S,S')$
  whose leaders lie in $\D_1$. The pairs in $\pairsbad{i}$ form an
  acyclic set in the auxiliary graph $\stageg{i}$. Consider the set of
  supernodes which occur as endpoints of the edges in $\pairsbad{i}$,
  and let $L_B$ be the set of terminals that are the leaders of these
  supernodes. Now pick a maximal set of these supernodes subject
  to the constraint that all of them lie in different moats in
  $\moatsone{i+1}$; i.e., no two of them lie in the same equivalence
  class. Let $L^\star_B \sse L_B$ be the set of terminals that are
  leaders of this maximal set. In the example given in Figure~\ref{fig:cost},
  $L_B=\{S_1, S_2, S_5, S_7, S_8, S_{10} \}$, and we could define $L^\star_B$ as
  $\{S_1,  S_5, S_7, S_8 \}$.
  By the fact that $\pairsbad{i}$ is an
  acyclic set, we get $|L^\star_B| \geq \countbad{i} + 1$.

  For any $s \in L_B$, let $M_s$ be the moat in $\moatsone{i}$
  containing $s$. By construction,
  \begin{OneLiners}
  \item For each pair $(S,S') \in \pairsbad{i}$, the leaders of $S$ and
    $S'$ lie in different moats in $\moatsone{i+1}$.
  \item For all $s \neq s' \in L_B^\star$, $s$ and $s'$ belong to different
    moats in $\moatsone{i+1}$.
  \end{OneLiners}
  Also note that terminals in distinct moats of $\moatsone{i+1}$ \emph{a
    fortiori} lie in distinct moats in $\moatsone{i}$. Now contract all
  the moats in $\moatsone{i}$ in the metric $\M$. Observe that all edges
  in $\forestone{\leq i-1}$ were already tight by the end of stage
  $i-1$, and hence get contracted by this operation.

  By Lemma~\ref{lem:moats}(b), for each $s \in L_B^\star$ the edges in
  $\forestone{\leq i-1} \cup \foresttwo{\leq i}$ connect moat $M_s$ to
  some another moat $M_{s'}$ for some $s' \in L_B$ (corresponding to the
  pair in $\pairsbad{i}$). Since we contracted the moats in
  $\moatsone{i}$, the edges in $\foresttwo{\leq i}$ connect $M_s$ to
  $M_{s'}$ in this contracted metric. But any two moats $M_s, M_{s'}$
  are at least $6 \cdot ( \cee^{i+1}-\cee^i)$ apart in this contracted metric
  (because these moats do not meet during stage~$i$ of the primal dual
  growing process, else $s, s'$ would share a moat in $\moatsone{i+1}$).
  Therefore, if we draw a ball of radius $3(\cee^{i+1}-\cee^i)$ around
  the moat $M_s$ in this contracted metric, this ball contains edges
  from $\foresttwo{\leq i}$ of length at least $3(\cee^{i+1}-\cee^i)$.

  Since these balls around the terminals in $L_B^\star$ are all
  disjoint, the total length of all edges in these balls is at least
  $$|L_B^\star| \cdot 3 \cdot (\cee^{i+1}-\cee^i) \geq 3 (\countbad{i}+1) \cdot (\cee^{i+1}-\cee^i).$$
  All these edges lie in moats in $\moatsone{i+1}$ but not within moats
  in $\moatsone{i}$, so we add these to the bound we get for
  $\foresttwo{\leq i-1} \cap E_{i-1}$ from the induction hypothesis and
  complete the inductive step.
\end{proof}

\begin{theorem}
  \label{thm:strict-cost}
  The cost of edges in $F_2$ is at most $6
  \gammatg$ times $\sum_{(s, \bar s) \in \D_2} \chi(\M, \D_1 \cup \D_2,
  (s, \bar s))$.
\end{theorem}
\begin{proof}
  By Lemma~\ref{lem:goodcost}, the cost of edges in $F_2$ is at least $3
  \sum_i \countbad{i} \cdot (\cee^{i+1} - \cee^i) \geq \frac32 \sum_i
  \countbad{i} \cdot \cee^{i+1}$. Out of these, the
  bad edges have total cost at most $\sum_i \countbad{i} \cdot \cee^{i+1}$,
  by~(\ref{eq:badcost}). So the cost of the good edges in $F_2$ is at
  least one third of the cost of all edges in $F_2$.

  However, observe that good edges correspond to pairs $(S,S') \in
  \pairsgood{i}$ for some $i$, i.e., this pair of supernodes was
  connected by these good edges in stage~$i$ of \timedglut, and
  moreover, at least one of the leaders of $S$ and $S'$ belong to a
  terminal pair in $\D_2$. By the construction of our cost shares, it
  follows the cost share of terminal pairs in $\D_2$ is at least
  $\frac1{2\gammatg}$ times the cost of the good edges. Hence the cost
  of the forest $F_2$ is at most $6 \gammatg$ times the cost share of
  terminal pairs in $\D_2$, proving the theorem.
\end{proof}
}

\subsubsection*{Acknowledgments} We thank R.\ Ravi for suggesting the
problem to us, for many discussions about it over the years, and for the
algorithm name. We also thank Chandra Chekuri, Jochen K\"onemann, Stefano
Leonardi, and Tim Roughgarden.

\ifstoc
 \bibliographystyle{abbrv}
 {\bibliography{bibonline}}
\else
 \bibliographystyle{alpha}
 {\small \bibliography{bibonline}}
\fi

\full{
\appendix

\section{The Lower Bound for (Paired) Greedy}
\label{sec:girth-lbd}

The (paired) greedy algorithm picks, at each time the closest
yet-unconnected source-sink pair $(s, \bar s)$ in the current graph, and
connects them using a shortest $s$-$\bar s$ path in the current graph. A
lower bound of $\Omega(\log n)$ for this algorithm was given by Chen,
Roughgarden, and Valiant~\cite{CRV10}. We repeat this lower bound here
for completeness.

Take a cubic graph $G = (V,E)$ with $n$ nodes and girth at least $c \log
n$ for constant $c$; see~\cite{Biggs-girth} for constructions of such
graphs. Fix a spanning tree $T$ of $G$, and let $E' = E \setminus E(T)$
be the non-tree edges. Set the lengths of edges in $E(T)$ to~$1$, and
the lengths of edges in $E'$ to $\frac{c}{2} \log n$.

Let $M$ be a maximal matching in $G' = (V,E')$; since $G$ and hence $G'$
has maximum degree $3$, this maximal matching has size at least
$\Omega(|E'|) = \Omega(n)$. The demand set $\D$ consists of the
matching~$M$.

We claim that the paired-greedy algorithm will just buy the direct edges
connecting the demand pairs. This will incur cost $|M| \cdot \frac{c}{2}
\log n = \Omega(n \log n)$. The optimal solution, on the other hand, is
to buy the edges of $T$, which have total cost $n-1$, which gives the
claimed lower bound of $\Omega(\log n)$.

The proof of the claim about the behavior of paired-greedy is by
induction. Suppose it is true until some point, and then demand $(u,v)$
is considered. By construction, $(u,v) \in M$. We can model the fact
that we bought some previous edges by zeroing out their lengths. Now the
observation is that for any path $P$ between $u$ and $v$ that is not the
direct edge $(u,v)$, if there are $k$ edges from $E'$ on this path, then
at most $\lceil k/2 \rceil$ of them can be zeroed out. Moreover, since
the girth is $g$, the number of edges on the path is at least $g-1$. So
the new length of the path $P$ is at least $(g-1-k) + \lfloor k/2
\rfloor \frac{c}{2} \log n \geq g-2$. This means the direct edge between
$u,v$ is still the (unique) shortest path, and this proves the claim and
the result.

\section{A Different Gluttonous Algorithm, and its Analysis}
\label{sec:projected}

In this section, we consider a slightly different analysis of the
gluttonous algorithm which does not rely on the faithfulness property
developed in Section~\ref{sec:near-optimal}. We then use this analysis
to show that a different gluttonous algorithm is also a constant-factor
approximation.

As before, let $\Fs$ denote an optimal solution to the \stf instance $\I
= (\M, \D)$, and let the trees in $\Fs$ be $\Ts_1, \Ts_2, \ldots,
\Ts_p$. A supernode $S$ is a subset of terminals -- note that we no
longer require that the supernodes formed during the gluttonous
algorithm should be contained in one of the trees of $\Fs$.  Let
$\clus{t}$ denote the clustering at the beginning of iteration $t$. So
$\clus{1}$ consists of singleton supernodes.

For each index $r$, $1 \leq r \leq p$, we can think of a new instance
$\I_r$ with set of terminals $V(\Ts_r)$ and metric $\M_r$ (the metric
$\M$ restricted to these terminals).  We now define the notion of {\em
  projected supernodes}. For each tree $\Ts_r$, we shall maintain a
clustering $\clusr{t}{r}$ corresponding to $\clus{t}$ -- this should be
seen as the {\em restriction} of the algorithm $\A$ to the instance
$\I_t$.  A natural way of defining this would be $S \cap V(\Ts_r)$ for
every $S \in \clus{t}$. But it turns out that we will really need a
refinement of the latter clustering. The reason for this is as follows:
if the algorithm $\A$ merges two active supernodes in the instance $\I$,
the intersection of the two supernodes with $V(\Ts_r)$ may be inactive
supernodes. But we do not want to combine two inactive supernodes in the
clustering $\clusr{t}{r}$.

We now define the clustering $\clusr{t}{r}$ formally. For a supernode
$S$, let $S_r$ denote $S \cap V(\Ts_r)$.  We say that $S_r$ is active if
there is some demand pair $(u, {\bar u})$ such that $u \in S_r$, but
${\bar u} \notin S_r$.  Let $\alive(S)$ denote the set of indices $r$
such that $S_r$ is active.  For each iteration $t$, we shall maintain
the following invariants.
\begin{itemize}
\item[(i)] The clustering $\clusr{t}{r}$ will be a refinement of the
  clustering $\{S_r: S \in \clus{t}\}$.
\item[(ii)] For each (active) supernode $S \in \clus{t}$ such that $S_r$
  is active, there will be exactly one active supernode $\act_r(S)
  \subseteq S_r$ in the clustering $\clusr{t}{r}$.
\end{itemize}
Initially, $\clusr{1}{r}$ is the clustering consisting of singleton sets
(it is easy to check that it satisfies the two invariants above because
$\clus{1}$ also consists of singleton sets). Suppose at iteration $t$,
the algorithm $\A$ merges supernodes $S'$ and $S''$ to a supernode
$S$. If $S'_r$ and $S''_r$ are both active (i.e., $r \in \alive(S') \cap
\alive(S'')$), then we replace $\act_r(S')$ and $\act_r(S'')$ by their
union in the clustering $\clusr{t}{r}$ to get
$\clusr{t+1}{r}$. Otherwise, $\clusr{t+1}{r}$ is same as
$\clusr{t}{r}$. It is easy to check that the above two invariants will
still be satisfied.  Further, observe that we only merge two active
supernodes of $\clusr{t}{r}$, and an inactive supernode never merges
with any other supernode.

Let $\delta_{t,r}$ denote the minimum over two active supernodes $S',
S''\in \clusr{t}{r}$ of $d_{\M_r/\clusr{t}{r}}(S', S'')$ (in case
$\clusr{t}{r}$ has only one active supernode, this quantity is
infinity). Fact~\ref{fact:gen} implies that $\delta_{t,r}$ is ascending
with $t$. Let $\delta_t$ denote the minimum over all $r$ of
$\delta_{t,r}$. Clearly, $\delta_t$ is also ascending with $t$.

Whenever the algorithm merges active supernodes $S'$ and $S''$, we will
think of merging the corresponding supernodes $\act_r(S')$ and
$\act_r(S'')$ in the instance $\I_r$ for each value of $r$ (provided $r
\in \alive(S') \cap \alive(S'')$).  Note that the latter merging cost
may not suffice to account for the merging cost of $S'$ and $S''$. For
example, suppose $S'=\{a\}, S''=\{b\}$, and $a$ and $b$ happen to be in
different trees in $\Fs$, say $\Ts_1$ and $\Ts_2$ respectively. If $S =
\{a,b\}$ is the new supernode, then $\act_1(S)=\act_1(S') = \{a\}$ and
$\act_2(S)=\act_2(S'') = \{b\}$. So the merging costs in the
corresponding instances is 0, but the actual merging cost is
positive. The main idea behind the proof is that we will pay for this
merging cost later when the algorithm merges the supernode containing
$a$ with some other supernode which has non-empty intersection with
$\Ts_1$ (and similarly for $b$).

In order to keep track of the unpaid merging cost, we associate a {\em
  charge} with each supernode (which gets formed during the gluttonous
algorithm $\A$) -- let $\charge(S)$ denote this quantity.  At the
beginning of every every iteration $t$ and supernode $S \in \clus{t}$,
we will maintain the following invariant:
\begin{eqnarray}
  \label{eq:chargeinv}
  \charge(S) \leq \left\{ \begin{array}{ll}  0 & \mbox{if $S$ is
        inactive at the beginning of iteration $t$. } \\
      (n(S)-1) \cdot \delta_t & \mbox{ otherwise,}  \end{array} \right.
\end{eqnarray}
where $n(S)$ denotes the number of indices $r$ such that $S_r$ is
active, i.e., $n(S) = |\alive(S)|$.  Now, we show how the quantity
$\charge(S)$ is updated. Initially (at the beginning of iteration 1), we
have a supernode $\{u\}$ for each terminal $u$.  The charge associated
with each of these supernodes is 0. Clearly, the invariant condition
above is satisfied.

Suppose the algorithm $\A$ merges supernodes $S'$ and $S''$ in iteration
$t$.  Let $S$ denote the new supernode. We define $\charge(S)$ as
\begin{eqnarray}
  \label{eq:charge}
  \charge(S) = \charge(S') + \charge(S'') - (2 |\alive(S') \cap
  \alive(S'')|-1) \delta_t.
\end{eqnarray}

We first prove that the invariant about charge of a supernode is satisfied.
\begin{claim}
  \label{cl:inv}
  The invariant conditions~(\ref{eq:chargeinv}) are always satisfied.
\end{claim}
\begin{proof}
  We prove this by induction on the iteration $t$. Suppose $\A$ merges
  (active) supernodes $S'$ and $S''$ iteration $t$, and let $S$ be the
  new supernode. If we consider a supernode $S_1$ other than $S$ in
  $\clus{t+1}$, then the statement follows easily. Note that
  $\charge(S_1)$ does not change.  If $S_1$ were inactive before
  iteration $t$, it will continue to be inactive at the end of iteration
  $t$ as well. If $S_1$ were active before iteration $t$, then it will
  continue to be active be at the beginning of iteration $(t+1)$ as
  well.  Further, $n(S_1)$ will not change (and is at least 1), and
  $\delta_{t+1} \geq \delta_t$. So, we now consider the case for $S$
  only.  We get
  \begin{align*}
    \charge(S) &=  \charge(S') + \charge(S'') -  (2 |\alive(S') \cap
    \alive(S'')|-1) \delta_t \\
    & \leq (n(S')-1) \cdot \delta_t + (n(S'')-1) \cdot \delta_t  -  (2
    |\alive(S') \cap \alive(S'')|-1) \delta_t  \\
    & \leq (n(S)-1) \cdot \delta_t,
  \end{align*}
  where the last inequality follows from the fact that $n(S) \geq n(S')
  + n(S'')-2| \alive(S') \cap \alive(S'')|$. Note that it is possible
  that $r \in \alive(S') \cap \alive(S'')$, but $r \notin \alive(S)$.
  However, $\alive(S') \backslash \alive(S''),$ $ \alive(S'') \backslash
  \alive(S') \subseteq \alive(S).$ Finally, observe that if $S$ is
  active, then $n(S) \geq 1$. This implies that $(n(S)-1) \cdot \delta_t
  \leq (n(S)-1) \cdot \delta_{t+1}$. Therefore, the invariant holds in
  this case. If $S$ becomes inactive, then $n(S)=0$, and so, $(n(S)-1)
  \cdot \delta_t \leq 0$. Again, the invariant continues to hold here.
\end{proof}

When the algorithm starts, $n(S)=1$ for all supernodes $S$ (consisting
of singleton terminals). So, $\charge(S)=0$ for all $S$ at time
$1$. When the algorithm terminates, $\charge(S) \leq 0$ for all
supernodes (since they are inactive). Therefore, adding the
inequality~({eq:charge}) over all iterations $t$, we get
\begin{eqnarray}
\label{eq:bounddel}
\sum_t \delta_t \leq 2 \delta_t \cdot \sum_t |\alive(S') \cap \alive(S'')| \leq 2 \sum_t \sum_r \one{r}{t} \cdot \delta_{t,r},
\end{eqnarray}
where $\one{r}{t}$ is the indicator variable which is 1 if and only if $r \in
\alive(S') \cap \alive(S'')$.  The following claim follows easily.

\begin{claim}
  \label{cl:mergingcost}
  The merging cost of $\A$ in iteration $t$ is at most $\delta_t$.
\end{claim}
\begin{proof}
  Suppose $\delta_t = \delta_{t,r}=d_{\M_r/\clusr{t}{r}}(\act_r(S'),
  \act_r(S''))$ for some active supernodes $S', S'' \in \clus{t}$.
  Observe that $d_{\M/\clus{t}}(S', S'')) \leq
  d_{\M_r/\clusr{t}{r}}(\act_r(S'), \act_r(S''))$.
\end{proof}

Finally, Corollary~\ref{cor:one-tree} applied to the instance $\I_r$
shows that
$$ \sum_t \one{r}{t} \cdot \delta_{t,r} \leq 48 \cdot \cost(\Ts_r). $$
Combining the above with Claim~\ref{cl:mergingcost} and
inequality~(\ref{eq:bounddel}), we get

\begin{lemma}
  \label{lem:r}
  The total merging cost of the gluttonous algorithm is at most $96$
  times that of the optimal solution.
\end{lemma}

\subsection{A Different Gluttonous Algorithm}

We apply the above result to analyze a slightly different version of the
gluttonous algorithm. As before, the algorithm will maintain a set of
supernodes, denoted by $\clus{t}$, at the beginning of iteration
$t$. Further, it will maintain an edge-weighted graph $\graph{t}$ with
vertex set being $\clus{t}$.  Initially, at time $t=1$, $\graph{1}$ is
just the initial graph $G$ on the vertices (i.e., terminals and Steiner
nodes). In iteration $t$, the algorithm picks the two active supernodes
with smallest distance between them, where distances are measured with
respect to the edge lengths in $\graph{t}$. Let $P_t$ be the
corresponding shortest path (all internal vertices of $P_t$ must be
inactive supernodes).  We contract all the edges of $P_t$ to get the
graph $\graph{t+1}$. The clustering $\clus{t+1}$ will be be same as
$\clus{t}$ except for the fact that the supernodes corresponding to
vertices of $P_t$ will get replaced by a new supernode which will be the
union of all these supernodes.

For the above algorithm, the proof of Theorem~\ref{thm:one-tree} does
not hold. More specifically, consider case \rom{1} in the proof where
$u$ and $v$ happen to be in the same tree in $\Fss$. Now, the algorithm
may contract a path with end-points containing $u$ and $v$
respectively. Further, the internal vertices of this path could
correspond to inactive supernodes, and so, the new supernode will
contain terminals outside this tree. But we cannot pay for the cost of
this path. However, we claim that the analysis of this section still
applies.

\begin{lemma}
  \label{lem:mod-glut}
  The above algorithm is a 96-approximation algorithm for the \stf problem.
\end{lemma}
\begin{proof}
  Let $\B$ denote the version of the gluttonous algorithm described
  above. It is no longer true that an inactive supernode does not merge
  with any other supernode. In each iteration, the algorithm may choose
  to merge two active supernodes along with several inactive
  supernodes. However, the clusterings $\clusr{t}{r}$ will respect the
  former property.

  The clusterings $\clusr{t}{r}$ will respect the invariants (i) and
  (ii) described above. When the algorithm $\B$ merges active supernodes
  $S'$ and $S''$ (along with possibly other inactive supernodes), and
  the index $r \in \alive(S') \cap \alive(S'')$, we merge $\act_r(S')$
  and $\act_r(S'')$ in the clustering $\clusr{t}{r}$ to get
  $\clusr{t+1}{r}$. It is easy to check that these two invariants are
  satisfied. The rest of the proof proceeds without any change. The key
  property that holds for this algorithm as well is that the merging
  cost (i.e., length of the shortest path) in each iteration is at most
  $\delta_t$ (defined in the proof above).
\end{proof}

}


\end{document}